\documentclass[sn-mathphys-num]{sn-jnl}

\usepackage{graphicx}
\usepackage{amsmath,amssymb,amsfonts}%
\usepackage{amsthm}%
\usepackage[title]{appendix}%
\usepackage{xcolor}
\usepackage{manyfoot}%
\usepackage[all]{xy}

\theoremstyle{plain}%
\newtheorem{theorem}{Theorem}
\newtheorem{lemma}{Lemma}
\newtheorem{proposition}{Proposition}%

\theoremstyle{definition}%
\newtheorem{example}{Example}%
\newtheorem{remark}{Remark}%

\newcommand{\bN}{\mathbb{N}}

\newcommand{\bR}{\mathbb{R}}

\newcommand{\bE}{\mathbb{E}}

\newcommand{\cH}{\mathcal{H}}

\newcommand{\grant}{Engineering and Physical Sciences Research Council [grant number EP/H031936/1], and the Biotechnology and Biological Sciences Research Council [grant numbers BB/L009579/1, BB/M021106/1]}

\begin{document}

\title[Probability of Beneficial Mutation and Crossover]{Analysis and Optimization of Probabilities of Beneficial Mutation and Crossover Recombination in a Hamming Space}


\author*[1]{\fnm{Roman V.} \sur{Belavkin}}\email{r.belavkin@mdx.ac.uk}

\affil*[1]{\orgdiv{Faculty of Science and Technology}, \orgname{Middlesex University}, \orgaddress{\street{The Burroughs}, \city{London}, \postcode{NW4 4BT}, \country{United Kingdom}}}

\abstract{Inspired by Fisher's geometric approach to study beneficial mutations, we analyse probabilities of beneficial mutation and crossover recombination of strings in a general Hamming space with arbitrary finite alphabet.  Mutations and recombinations that reduce the distance to an optimum are considered as beneficial.  Geometric and combinatorial analysis is used to derive closed-form expressions for transition probabilities between spheres around an optimum giving a complete description of Markov evolution of distances from an optimum over multiple generations.  This paves the way for optimization of parameters of mutation and recombination operators.  Here we derive optimality conditions for mutation and recombination radii maximizing the probabilities of mutation and crossover into the optimum.  The analysis highlights important differences between these evolutionary operators.  While mutation can potentially reach any part of the search space, the probability of beneficial mutation decreases with distance to an optimum, and the optimal mutation radius or rate should also decrease resulting in a slow-down of evolution near the optimum.  Crossover recombination, on the other hand, acts in a subspace of the search space defined by the current population of strings.  However, probabilities of beneficial and deleterious crossover are balanced, and their characteristics, such as variance, are translation invariant in a Hamming space, suggesting that recombination may complement mutation and boost the rate of evolution near the optimum.
}

\keywords{Mutation, Crossover, Recombination, Evolutionary Algorithm, Optimal Parameter Control}

\pacs[MSC Classification]{%
05B25, 
68W20, 
68T05, 
60C05, 
68R05, 
68R15 
}

\maketitle

\tableofcontents

\section*{Notation}

\begin{tabular}{lp{.8\textwidth}}
  $\bN$ & the set of natural numbers $\{1,2,3,\ldots\}$\\
  $\bR$ & the field of real numbers\\
  $l\in\bN$ & `length' of tuples or strings\\
  $\alpha\in\bN$ & size of a finite alphabet\\
  $\{1,\ldots,\alpha\}^l$ & the set of all $\alpha^l$ strings of length $l$ over alphabet of size $\alpha$\\
  $\bR^l$ & $l$-dimensional real vector space\\
  $d_E$ & Euclidean metric on $\bR^l$\\
  $d_H$ or $d$ & Hamming metric on $\{1,\ldots,\alpha\}^l$\\
  $\cH_\alpha^l$ & Hamming space --- set $\{1,\ldots,\alpha\}^l$ with the Hamming metric\\
  $x,y,z$ & points in $\cH_\alpha^l$, which are $l$-tuples or strings $x=(x_1,\ldots,x_l)$\\
  $i$, $j$ & positions in strings, such as $x=(x_1,\ldots,x_i,x_j,\ldots,x_l)$\\
  $\top$ &  top or greatest string in $\cH_\alpha^l$ with respect to some preference relation $\lesssim$ on $\cH_\alpha^l$\\
  $k$, $m$, $n$ & values of Hamming distance from $\top$, such as $d(\top,x)=n$\\
  $r$ & \emph{mutation radius} $d(x,y)=r$ in the context of mutation or \emph{recombination radius} in the context of recombination (the number of letters substituted)\\ 
  $h$ & Hamming distance $d(x,y)=h$ between two parent strings in crossover recombination and referred to as \emph{recombination capacity}\\
  $S(x,r)$ & the sphere of radius $r$ around $x$ $\{y:d(x,y)=r\}$\\
  $B(x,r)$ & the closed ball of radius $r$ around $x$ $\{y:d(x,y)\leq r\}$\\
  $\mu$ & mutation rate\\
  $\nu$ & recombination rate\\
  $P\{\cdot\}$ & probability measure\\
  $P(n)$ & probability mass function equal to $P\{d(\top,x)=n\}$\\
  $\bE_P\{n\}$ & the expected value of random variable with respect to measure $P$\\
  $\sigma_P^2\{n\}$ & the variance of random variable with respect to measure $P$\\
\end{tabular}

\section{Introduction}
\label{sec:intro}

Natural evolution can be viewed as a search for an optimal genotype $\top$ (top) in the space $\{1,\ldots,\alpha\}^l$ of all genetic codes of finite alphabet $\{1,\ldots,\alpha\}$ of size $\alpha\in\bN$ and length $l\in\bN$.  Optimality can be defined by some fitness function $f:\{1,\ldots,\alpha\}^l\to\bR$ maximized at $\top$.  Some mathematicians, however, simplified the analysis by replacing fitness $f(x)$ of a genotype with its distance $d(\top,x)$ from $\top$.  For example, Roland Fisher \cite{Fisher30} used Euclidean space $\bR^l$ of $l$ traits to represent species by vectors of $l$ traits and Euclidean distance $d_E(\top,x)$ from an optimum to represent (negative) fitness of $x$.  This simplification allowed him to analyse the probability of beneficial mutations, which in this geometric model meant that mutation of $x$ resulted in an offspring $y$ closer to the optimum: $d_E(\top,y)\leq d_E(\top,x)$.  Fisher's famous result was that beneficial mutations are always more rare than deleterious, and that the only way to equalize their chances is to minimize the mutation radius $d_E(x,y)=r$.  This result follows from the geometry of Euclidean space, where every closed ball $B(\top,n)=\{x\in\bR^l:d_E(\top,x)\leq n\}$ around $\top$ is compact (and has finite volume), while its complement is always unbounded.  Thus, a random mutation of $x$ with $d_E(\top,x)=n$ in all directions by radius $d(x,y)=r$ should more likely end outside the ball $B(\top,n)$ and further from the optimum resulting in a deleterious mutation.

The discovery of DNA and RNA molecules lead to the realization that mutations occur on the level of genetic codes, which are better represented as strings $x=(x_1,\ldots,x_l)$ of length $l\in\bN$ over some finite alphabet $\{1,\ldots,\alpha\}\ni x_i$.  Thus, Fisher's theory of beneficial mutations had to be reconsidered for spaces of strings with alphabets of arbitrary size $\alpha\in\bN$ and variable lengths $l\in\bN$ \cite{Belavkin_etal11:_ecal11,Belavkin11:_itw11,Belavkin11:_dyninf,Belavkin11:_qbic11}.  Furthermore, this geometric approach (i.e. replacing fitness with distance) had limited appeal for practical applications, because distances to an optimum are usually not known.  However, the values $f(x)$ of a fitness function can often provide some information about the distance $d(\top,x)$, and the correlation between fitness and distance has been discussed in the literature, for example, as a measure of problem difficulty \cite{Jones-Forrest95:_fdc,Poli-Galvan12}.  Various notions of monotonicity of fitness landscapes have been defined and proven to hold in a broad class of landscapes if they are continuous at least at the optimum $\top$ (see Theorem~1 in \cite{Belavkin_etal16:_jomb}).  While it is always possible to construct counter examples, fitness landscapes in real-world applications or biology often exhibit some forms of monotonicity around optimum, as was demonstrated in \cite{Belavkin_etal16:_jomb} for 115 complete landscapes of transcription factor bindings \cite{Badis09}.

The generalization of Fisher's geometric model of beneficial mutations to spaces of strings with alphabets of arbitrary size $\alpha\in\bN$ was used to derive several optimal mutation rate control functions \cite{Belavkin_etal11:_ecal11,Belavkin11:_itw11,Belavkin11:_dyninf,Belavkin11:_qbic11}.  They showed that optimal mutation rates should have a decreasing relation to fitness in monotonic fitness landscapes \cite{Belavkin_etal16:_jomb}.  These theoretical predictions lead to the discovery of mutation rates plasticity first in \emph{e. coli} \cite{eids_nature14} and then in other microbes and potentially across all domains of life \cite{Krasovec_etal17:_plos}.  The role of \emph{quorum sensing} in this phenomenon and the relation of population density (as a fitness proxy) and stress to mutation rate \cite{Krasovec18:_isme} suggest a broad scope for applications in many areas including antimicrobial resistance (AMR).  Another potential area of applications of this geometric approach is operational research, where many nature-inspired algorithms~\cite{Yang10} are used to solve complex combinatorial optimization problems.

Evolutionary algorithms, such as genetic algorithms (GA), encode candidate solutions by finite length strings $x=(x_1,\ldots,x_l)$ with letters from a finite alphabet $x_i\in\{1,\ldots,\alpha\}$, and operators of selection, mutation and recombination are applied iteratively to search the space $\{1,\ldots,\alpha\}^l$ \cite{Back93}.  Mutation is a random substitution of some letters in the parent string by any of the $\alpha-1$ letters from the alphabet.  Recombination, on the other hand, is a substitution of some letters in one parent string by the letters from another string (e.g. in the corresponding positions for crossover recombination).  Thus, mutation searches across the entire space $\{1,\ldots,\alpha\}^l$, while recombination can only search in a subspace defined by the current population.  However, recombination of different strings makes some directions of the search more likely (i.e. a kind of pseudo-gradient).

Many heuristics have been identified to improve the search efficiency by finding optimal settings or optimal controls of certain parameters, such as the mutation rate.  In particular, one popular heuristic is to set the mutation rate to $\mu=1/l$, where $l$ is the string length \cite{Ochoa02}.  Other works showed the advantage of using a variable mutation rate that may depend on time or fitness of individuals \cite{Fogarty89,Yanagiya93,Back93,Srinivas-Patnaik94,Braga-Aleksander94,Eiben_etal99,Falco_etal02}.  Many of these works considered only binary codes ($\alpha=2$), because their combinatorics is more tractable.  More recent studies in the theory of evolutionary algorithms have considered arbitrary finite alphabets and self-adjusting mutation rates \cite{Doerr_etal11,Doerr-Pohl12,Doerr_etal18}.

Different heuristic recombination operators have also been employed, such as one-point crossover or a uniform crossover operators, and its important role in maintaining diversity has been recognized \cite{Koetzing11:_crossover,Dang18:_crossover}.  While there are many other types of recombination operators considered in the literature, including mixtures of codes from more than two parents \cite{Moraglio07}, this paper will only consider crossover between two parent strings.  Even in this basic case, however, combinatorial analysis of crossover is more challenging than that for mutation, because it involves more points and more parameters.  Many studies have considered recombination only for binary codes \cite{Srinivas-Patnaik94,Eremeev99,Eremeev08,Eremeev2011,Eremeev-Kovalenko14:_1,Eremeev-Kovalenko14:_2}.

The analysis of evolutionary operators for codes with alphabets of size $\alpha>2$ should have a broader scope of applications not only in the context of DNA or RNA molecules with $\alpha=4$, but also for larger alphabets, such as $\alpha=22$ for the number of canonical amino acids.  In addition, some recombination operators substitute entire substrings of length $r$ (e.g. $r=\lfloor l/2\rceil$ in one-point crossover).  Therefore, recombination can be considered acting on the space of strings $\{1,\ldots,\alpha^r\}^{l/r}$ (i.e. alphabet of size $\alpha^r$).

This work develops a geometric approach to evolution of strings in Hamming spaces extending it from mutation to crossover recombination.  In the next section, we start by reviewing some of the basic properties of a Hamming space and formulate the problem of finding the probability of mutation onto a sphere of a given radius.  Its closed-form solution is given in Theorem~\ref{th:mutation}, which has been previously presented in \cite{Belavkin_etal11:_ecal11,Belavkin11:_itw11,Belavkin11:_dyninf,Belavkin11:_qbic11,Belavkin_etal16:_jomb}.  These results about mutation are included here not only for completeness, but also because they are used in the analysis of recombination in Section~\ref{sec:recombination}, and in particular Lemma~\ref{lm:spheres} for intersection of spheres.  We also derive new formulae for the conditional expected value and variance of Hamming distance after mutation.  In Section~\ref{sec:recombination}, we formulate analogous problem for probability of crossover recombination onto a sphere in a Hamming space, and then derive closed-form solution in Theorem~\ref{th:recombination}.  As with mutation, we also derive new formulae for the expected value and variance of distance after recombination.  We conclude each section by analysing the effects of parameters on probabilities of beneficial mutation and recombination and deriving optimality conditions for mutation and crossover recombination into optimum.  We discuss how our results open new possibilities for a long-term analysis and optimization of mutation and recombination operators.

\section{Mutation}
\label{sec:mutation}

\subsection{Mutation in a Hamming space}

Consider the space $\{1,\ldots,\alpha\}^l$ of strings (or codes) of length $l\in\bN$ and finite alphabet of size $\alpha\in\bN$, and let us equip it with the Hamming metric counting the number of different letters:
\begin{equation}
d(x,y)=|\{i\in\{1,\ldots,l\}:x_i\neq y_i\}|=\sum\limits_{i=1}^l(1-\delta_{x_iy_i})\,,\quad
\delta_{x_iy_i}=
\begin{cases}
  1 & \text{if $x_i=y_i$}\\
  0 & \text{otherwise}
\end{cases}\,.
\label{eq:hamming}
\end{equation}
This metric space is referred to as Hamming space and denoted by $\cH_\alpha^l$.  The definition of Hamming metric~(\ref{eq:hamming}) as the sum of elementary distances $1-\delta_{x_iy_j}$ leads to the following useful result.

\begin{lemma}[Mean and variance of Hamming distance]
  Let $\cH_\alpha^l$ be a Hamming space, and let $P:2^{\cH_\alpha^l\times\cH_\alpha^l}\to[0,1]$ be a joint probability distribution.  Then the expected value and variance of the Hamming distance $d(x,y)$ between pairs of strings $x$, $y\in\cH_\alpha^l$ are
  \begin{align}
    \bE_P\{d(x,y)\}&=l\langle P_i\rangle\,, \label{eq:e-lp}\\
    \sigma_P^2\{d(x,y)\}&=l\langle P_i\rangle+l(l-1)\langle P_{ij}\rangle-(l\langle P_i\rangle)^2\,,\label{eq:var-lp}
  \end{align}
  where
  \begin{align*}
    \langle P_i\rangle&:=\frac{1}{l}\sum_{i=1}^l\bE_P\{1-\delta_{x_iy_i}\}=\mathbb{P}\{x_i\neq y_i\}\,,\\
    \langle P_{ij}\rangle&:=\frac{1}{l(l-1)}\sum_{i=1}^l\sum_{\substack{j=1\\j\neq i}}^l\bE_P\{(1-\delta_{x_iy_i})(1-\delta_{x_jy_j})\}=\mathbb{P}\{x_i\neq y_i \wedge x_j\neq y_j\mid i\neq j\}\,.
  \end{align*}
  are respectively the average probability of non-identical letters at positions $i\in\{1,\ldots,l\}$ and the average joint probability of non-identical letters at two different positions $i$ and $j\neq i$ under the distribution $P$.
  \label{lm:moments}
\end{lemma}

See Appendix~\ref{sec:proof-e-var} for the proof.  The average probabilities $\langle P_i\rangle:=\mathbb{P}\{x_i\neq y_i\}$ and $\langle P_{ij}\rangle:=\mathbb{P}\{x_i\neq y_i \wedge x_j\neq y_j\mid i\neq j\}$ can often be estimated from additional information about distances of strings $x$ and $y$.  For example, assume a joint distribution $P(x,y)$ such that $d(x,y)=n$ for all pairs of strings such that $P(x,y)>0$ (i.e. exactly $n$ letters $x_i\neq y_i$).  Then $\langle P_i\rangle=n/l$ and  $\langle P_{ij}\rangle=(n/l)[(n-1)/(l-1)]$.  Substituting these probabilities into equations~(\ref{eq:e-lp}) and (\ref{eq:var-lp}) gives the desired result $\bE_P\{d(x,y)\}=n$ and $\sigma_P^2\{d(x,y)\}=0$.  More interesting and useful formulae will be obtained in Proposition~\ref{pr:mutation} for mutation and Proposition~\ref{pr:recombination} for crossover.

The geometry of Hamming space is different from that of the Euclidean space $\bR^l$ employed by Fisher \cite{Ahlswede-Katona77}.  In particular, $\cH_\alpha^l$ is finite, has finite diameter $l$, and every point $x\in\cH_\alpha^l$ has $(\alpha-1)^l$ diametric opposite points $\neg x$ (i.e. such that $d(x,\neg x)=l$).  The number of elements in a sphere $S(\top,n):=\{x\in\cH_\alpha^l:d(\top,x)=n\}$ of radius $n$ around $\top$ is
\[
  |S(\top,n)|=(\alpha-1)^n\binom{l}{n}\,.
\]
The number of elements in a closed ball $B(\top,n):=\{x\in\cH_\alpha^l:d(\top,x)\leq n\}$ is the sum $\sum_{k=0}^n|S(\top,k)|$ for all the spheres it contains.   The complement $H_\alpha^l\setminus B(\top,n)$ is the union of all balls $B(\neg\top,l-n)$ around $(\alpha-1)^l$ diametric opposite points $\neg\top$ (see \cite{Ahlswede-Katona77} for details and many other properties of Hamming space).  The number of elements in the complement $H_\alpha^l\setminus B(\top,n)$ is the sum $\sum_{k=l-n}^l|S(\top,k)|$, and it may contain fewer elements than the ball itself, unlike in the Euclidean space.

The `equator' of a Hamming space is the radius equal to $\lfloor l(1-1/\alpha)\rceil$ (here $\lfloor\cdot\rceil$ denotes the nearest integer), which corresponds to the median of the binomial distribution $P(n;l,p)$ with parameter $p=1-1/\alpha$.  Indeed, under a uniform distribution $P_0(x)=\alpha^{-l}$ of strings $x\in\cH_\alpha^l$, the probability $P_0(n)$ of distances $d(\top,x)=n$ from $\top$ (or from any other point) is
\[
P_0(x)=\frac{1}{\alpha^l}\quad\implies\quad P_0(n):=P_0\{x\in S(\top,n)\}=\frac{|S(x,n)|}{\alpha^l}=\frac{(\alpha-1)^n}{\alpha^l}\binom{l}{n}\,,
\]
which can be written as the binomial distribution $P_0(n)=\binom{l}{n}(1-1/\alpha)^n(1/\alpha)^{l-n}$.  Its expected value and variance are
\[
  \bE_{P_0}\{n\}=l(1-1/\alpha)\,,\qquad \sigma^2_{P_0}(n)=l(1-1/\alpha)(1/\alpha)\,,
\]
which can also be obtained using formulae~(\ref{eq:e-lp}) and (\ref{eq:var-lp}) with $\langle P_i\rangle=1-1/\alpha$ and $\langle P_{ij}\rangle=(1-1/\alpha)^2$.  The median is the nearest integer of the above expected value.  For alphabets of size $\alpha>2$ the distribution of distances is skewed towards the end of the range $[0,l]$.

\begin{figure}[!t]
\centering
\setlength{\unitlength}{1mm}
\begin{picture}(50,40)
\linethickness{.075mm}
\put(34,16){\circle*{1}}
\put(34,14){$x$}
\put(34,16){\vector(0,1){7}}
\put(34,23){\circle*{1}}
\put(35,23.5){$y$}
\put(34.5,19){$r$}
\put(34,16){\circle{14}}
\put(5,35){\circle*{1}}
\put(2,35){$\top$}
\put(5,35){\vector(3,-2){28}}
\put(18,22){$n$}
\qbezier(5,5)(40,5)(40,40)
\put(5,35){\vector(4,-1){30}}
\put(22,32){$m$}
\qbezier(5,8)(37,8)(37,40)
\end{picture}
\caption{Mutation of string $x\in S(\top,n)$ into $y\in S(\top,m)$ by substitution of $r=d(x,y)$ letters.  The number of strings in the intersection $S(\top,m)\cap S(x,r)$ defines the geometric probability $P(m\mid n,r)$ (\ref{eq:p-mutation-def}).}
\label{fig:mutation}
\end{figure}
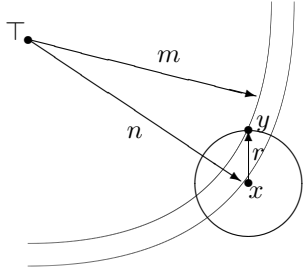

Asexual reproduction of species corresponds to a transformation $x\mapsto y$ of their genetic codes due to a random substitution of $r\in[0,l]$ letters --- a process which we shall generally refer to as \emph{mutation}.  The resulting distance $d(x,y)=r$ from the parent string in this context is referred to as the \emph{mutation radius} shown on Figure~\ref{fig:mutation}.  If distance $d(\top,\cdot)$ from the optimum $\top$ is taken as a model of (negative) fitness, then \emph{beneficial} mutation is a transition from sphere $S(\top,n)\ni x$ onto sphere $S(\top,m)\ni y$ of a smaller radius $m<n$, as shown on Figure~\ref{fig:mutation}.  Mutation is \emph{neutral} if $m=n$, and \emph{deleterious} if $m>n$.  

\begin{example}[Mutation]
Let $\top=(\mathrm{AAAAA})\in\cH_3^5$ and consider string $x=(\mathrm{BBBAA})$ mutating into $y=(\mathrm{BACBA})$ by a substitution of the second, third and fourth letters:
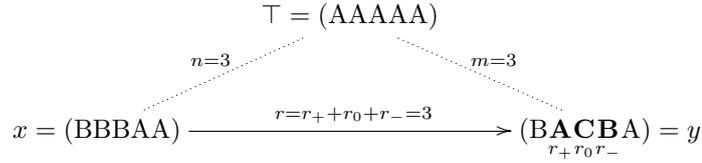
\[
  \xymatrix{&\top=(\mathrm{AAAAA})\ar@{.}[ld]_{n=3}\ar@{.}[rd]^{m=3}&\\
    x=(\mathrm{BBBAA})\ar[rr]^{r=r_++r_0+r_-=3}&&(\mathrm{B}\underset{r_+}{\mathbf{A}}\underset{r_0}{\mathbf{C}}\underset{r_-}{\mathbf{B}}\mathrm{A})=y}
\]
Thus, the mutation radius is $d(x,y)=3$, and the mutation is neutral, because $d(\top,x)=d(\top,y)=3$.  Notice that there was $r_+=1$ beneficial, $r_0=1$ neutral and $r_-=1$ deleterious substitution.  A substitution of three letters in $x=(\mathrm{BBBAA})$ may also result in string $z=(\mathrm{BAACA})$, which is closer to the optimum, $d(\top,z)=2$ (i.e. beneficial mutation).
\label{ex:mutation}
\end{example}

Henceforth we shall denote by $r_+$, $r_-$ and $r_0$ the numbers of beneficial, deleterious and neutral substitutions respectively.  These numbers add up to the mutation radius $r=d(x,y)$, while the difference $r_+-r_-$ is equal to the difference $n-m$ of distances:
\begin{align}
  r_++r_-+r_0=r \label{eq:rm-sum}\,,\\
  r_+-r_-=n-m \label{eq:rm-difference}\,.
\end{align}
If string $x\in S(\top,n)$ mutates into $y\in S(\top,m)$, then the range of the mutation radius is defined by the triangle inequalities:
\[
  |n-m|\leq r\leq n+m
\]
At the extreme values $r=|n-m|$ or $r=n+m$ of the mutation radius, there are no neutral substitutions.  Indeed, for the maximum value $r=n+m$ there are exactly $r_+=n=d(\top,x)$ beneficial and $r_-=m=d(\top,y)$ deleterious substitutions, so that $r_0=r-r_+-r_-=0$. For the minimum value $r=|n-m|$ there are $r_+=\max\{0,n-m\}$ beneficial and $r_-=\max\{0,m-n\}$ deleterious substitutions.  Then equation~(\ref{eq:rm-sum}) and $r=|n-m|=\max\{n-m,m-n\}$ imply $r_0=0$.  In both extreme cases we also have $r_-=r-r_+$ and $r_+=\frac12(r+n-m)$ (if the latter is integer).  Clearly, neutral substitutions are impossible for binary strings ($\alpha=2$).
 
\subsection{Evolutionary dynamics under mutation}

A random mutation $x\mapsto y$ corresponds to some transition probability $P(y\mid x)$, where $x$ is the `parent', and $y$ is its `offspring', and it induces a transformation of distribution $P_t(x)$ of the parent codes into the distribution $P_{t+1}(y)$ of the offspring codes:
\[
P_{t+1}(y)=\sum_{x\in\cH_\alpha^l} P(y\mid x)\,P_t(x)\,.
\]
The corresponding distance distributions $P_t(n):=P_t\{x\in S(\top,n)\}$ and $P_{t+1}(m):=P_{t+1}\{y\in S(\top,m)\}$ are transformed as well:
\[
P_{t+1}(m)=\sum_{n=0}^l P(m\mid n)\,P_t(n)\,.
\]
Here, $P(m\mid n)$ is the transition probability between spheres around $\top$ due to mutation:
\[
P(m\mid n):=P\{y\in S(\top,m)\mid x\in S(\top,n)\}\,.
\]
If $P(m\mid n)$ is time invariant, then the linear operator
\begin{equation}
  M(\cdot)=\sum_{n=0}^lP(m\mid n)\,(\cdot)
  \label{eq:operator-mutation}
\end{equation}
acting on distributions $P_t(n)$ of distances $d(\top,x)=n\in[0,l]$ generates the entire evolution $\{P_t\}_{t\geq0}$ due to mutation as $P_{t+s}=M^sP_t$.  This can be used in simulations to analyse the effects of mutation and adaptation over several generations.

The transition probability $P(m\mid n)$ can be factorized in the following way:
\begin{equation}
  P(m\mid n)=\sum\limits_{r=0}^l P(m\mid n,r)\,\underbrace{P(r\mid n)}_{\text{Mutation}}\,,
  \label{eq:mutation-factorization}
\end{equation}
where $P(r\mid n):=P\{y\in S(x,r)\mid x\in S(\top,n)\}$ is the probability of mutation radius $r\in[0,l]$ conditioned on distance $n=d(\top,x)$.  This probability can be determined from the mutation operator.

\begin{example}[Point mutation]\label{ex:point-mutation}
In a simple point mutation, each letter $x_i$ is substituted independently with probability $\mu\in[0,1]$ called the \emph{mutation rate} (i.e. $\mu$ is fixed for all $i\in\{1,\ldots,l\}$).  Each letter $x_i$ can be substituted to any of the $\alpha-1$ letters $y_i$ with uniform probability $1/(\alpha-1)$.  In this case, the probability that $r\in[0,l]$ letters are substituted has binomial distribution:
\begin{equation}
  P_\mu(r\mid n)=\binom{l}{r}\mu^r(n)[1-\mu(n)]^{l-r}\,.
  \label{eq:simple-mutation}
\end{equation}
Note that the mutation rate $\mu$ may depend on the distance $d(\top,x)=n\in[0,l]$ of the parent string from the optimum.  The mutation operator~(\ref{eq:operator-mutation}) in this case is parameterized by the mutation rate control function $\mu(n)$:
\[
  M_{\mu(n)}(\cdot)=\sum_{n=0}^l\left[\sum_{r=0}^l P(m\mid n,r)P_\mu(r\mid n)\right]\,(\cdot)\,.
\]
\end{example}

\subsection{Geometric probability of mutation onto a sphere}

The probability $P(m\mid n,r)$ in factorization~(\ref{eq:mutation-factorization}) is independent of the mutation operator, and it represents a purely geometric problem depicted on Figure~\ref{fig:mutation}:
\begin{equation}
  P(m\mid n,r):=P\{y\in S(\top,m)\mid x\in S(\top,n),d(x,y)=r\}\,.
  \label{eq:p-mutation-def}
\end{equation}
Fisher considered this geometric probability in Euclidean space $\bR^l$ \cite{Fisher30}.  For a Hamming space $\cH_\alpha^l$, this problem was considered in \cite{Belavkin_etal11:_ecal11,Belavkin11:_itw11,Belavkin11:_dyninf,Belavkin11:_qbic11,Belavkin_etal16:_jomb}, and here we review its solution, because it will be useful for the analysis of recombination in Section~\ref{sec:recombination}.  One can see from Figure~\ref{fig:mutation} that probability~(\ref{eq:p-mutation-def}) depends on the number of strings in the intersection of spheres $S(\top,m)$ and $S(x,r)$.

\begin{lemma}[Intersection of spheres \cite{Belavkin_etal11:_ecal11,Belavkin11:_itw11,Belavkin11:_dyninf,Belavkin11:_qbic11,Belavkin_etal16:_jomb}]
  The number of elements in the intersection $S(\top,m)\cap S(x,r)$ of spheres around points $\top$, $x\in\cH_\alpha^l$ with $d(\top,x)=n$ is
  \begin{equation}
    \Bigl|S(\top,m)\cap S(x,r)\Bigr|_{d(\top,x)=n}
    =\sum\limits_{r_+=0}^r(\alpha-2)^{r_0}\binom{n-r_+}{r_0}(\alpha-1)^{r_-}\binom{l-n}{r_-}\binom{n}{r_+}\,,
    \label{eq:intersection-sphere}
  \end{equation}
  where $r_+\in[0,r]$, and indices $r_-\geq0$, $r_0\geq0$ satisfy the equations:
  \[
    r_-=r_+-(n-m)\,,\qquad r_0=r-2r_++(n-m)\,.
  \]
  Distances $n=d(\top,x)$, $m=d(\top,y)$ and $r=d(x,y)$ must satisfy the triangle inequalities:
  \[
    |n-r|\leq m\leq n+r\,.
  \]
  Otherwise, the number is zero.
  \label{lm:spheres}
\end{lemma}

See Appendix~\ref{sec:proof-spheres} for the proof.

\begin{remark}
  The summation in equation~(\ref{eq:intersection-sphere}) is shown across all $r_+\in[0,r]$, but it is important to check also that $r_-=r_+-(n-m)\geq0$ and $r_0=r-2r_++(n-m)\geq0$.  The triangle inequalities $|n-m|\leq r\leq n+m$ imply the following bounds $\max\{0,n-m\}\leq r_+\leq\frac12(r+n-m)\leq r$, which can be used for a more efficient implementation.
  \label{rem:range-r}
\end{remark}

Formula~(\ref{eq:intersection-sphere}) for the binary case $\alpha=2$ was previously analysed in \cite{Braga-Aleksander94} (see also \cite{Back93,Eremeev99}).  The solution for arbitrary alphabets was first given in \cite{Belavkin_etal11:_ecal11} (see also \cite{Belavkin11:_itw11,Belavkin11:_dyninf,Belavkin11:_qbic11}).   We now have all information required to find probability~(\ref{eq:p-mutation-def}).

\begin{figure}[t]
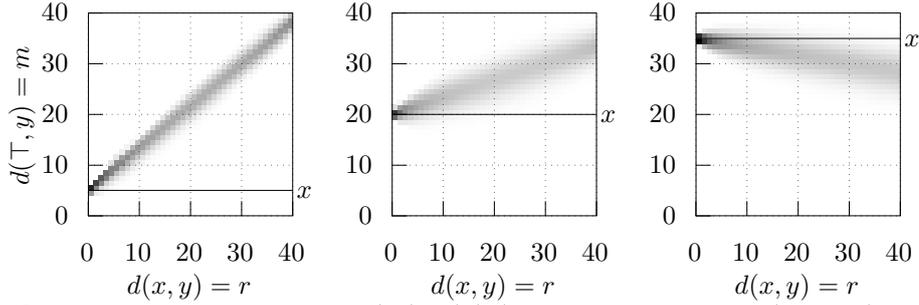

  \centering
  \input{p-mutation-r-5-40}
  \input{p-mutation-r-20-40}
  \input{p-mutation-r-35-40}
  \caption{Dependency of the probability $P(m\mid n,r)$ (\ref{eq:p-mutation}) on the mutation radius $r$ (abscissae) in space $\cH_4^{40}$.  Ordinates show the resulting distance $d(\top,y)=m$ after mutation.  Three charts correspond to three distances $d(\top,x)=n\in\{5,20,35\}$ of the parent string.  Grayscale represents probability $P\in[0,1]$.}
  \label{fig:p-mutation-r}
\end{figure}

\begin{theorem}[Geometric probability of mutation onto a sphere \cite{Belavkin_etal11:_ecal11,Belavkin11:_itw11,Belavkin11:_dyninf,Belavkin11:_qbic11}]
  The probability $P(m\mid n,r)$ that a substitution of $r\in[0,l]$ letters in string $x\in S(\top,n)\subset\cH_\alpha^l$ results in string $y\in S(\top,m)$ is
\begin{equation}
  P(m\mid n,r)=\frac{\sum\limits_{r_+=0}^r(\alpha-2)^{r_0}\binom{n-r_+}{r_0}(\alpha-1)^{r_-}\binom{l-n}{r_-}\binom{n}{r_+}}{(\alpha-1)^r\binom{l}{r}}
  \label{eq:p-mutation}
\end{equation}
with $r_+\in[0,r]$ and the numbers $r_-\geq0$, $r_0\geq0$ defined by the equations
\[
  r_-=r_+-(n-m)\,,\qquad r_0=r-2r_++(n-m)\,.
\]
The probability is zero if the triangle inequalities $|n-r|\leq m\leq n+r$ are not satisfied.
\label{th:mutation}
\end{theorem}

\begin{proof}
  The probability is given by the proportion of strings in sphere $S(x,r)$ that are also in the sphere $S(\top,m)$:
  \[
  P(m\mid n,r)=\frac{|S(\top,m)\cap S(x,r)|_{d(\top,x)=n}}{|S(x,r)|}\,.
  \]
 The number in the intersection is given by equation~(\ref{eq:intersection-sphere}), and the number of elements in $S(x,r)$ is $(\alpha-1)^r\binom{l}{r}$.
\end{proof}

\begin{remark}
  The proof above makes an implicit assumption about a uniform distribution within subsets (spheres), which is common if no other information about the distribution is given (i.e. the principle of insufficient reason).
\end{remark}

\begin{example}[Binary case $\alpha=2$]
  For binary strings, formula~(\ref{eq:p-mutation}) reduces to:
  \begin{equation}
  P(m\mid n,r)=\frac{\binom{l-n}{r - r_+}\binom{n}{r_+}}{\binom{l}{r}}\,,
  \label{eq:p-mutation-2}
\end{equation}
where $r_+=\frac12(r+n-m)$ must be non-negative integer (otherwise, the probability is zero).  Note that the right-hand-side of equation~(\ref{eq:p-mutation-2}) is the hypergeometric distribution $\mathbb{P}\{X=r_+\}$ of $r_+\in[0,r]$ if it is considered as a random variable.    The above formula is also valid for $\alpha>2$ when the mutation radius $d(x,y)=r$ is minimized ($r=|n-m|$) or maximized ($r=n+m$), because there are no neutral substitutions in these cases.
\label{ex:mutation-2}
\end{example}

Conditional probability~(\ref{eq:p-mutation}) was implemented in a digital computer using Common Lisp programming language, and Figure~\ref{fig:p-mutation-r} illustrates its dependency on parameters $n$ and $r$ in Hamming space $\cH_4^{40}$ ($\alpha=4$, $l=40$).  Three charts correspond to three values of distance $n\in\{5,20,35\}$ of the parent string.  Abscissae show mutation radii $r=d(x,y)$, while ordinates show the resulting distances $m=d(\top,y)$ of offsprings after mutation.  The grayscale represents different values of probability $P(m\mid n,r)$ with white corresponding to $P=0$ and black to $P=1$.  One can see that the mutation radius has different effects on the probability depending on whether the parent's distance $n=d(\top,x)$ is less or greater than the `equator' $l(1-1/\alpha)$: for $n<l(1-1/\alpha)$ higher mutation radius makes larger distances $m=d(\top,y)$ more likely, but the effect reverses for $n>l(1-1/\alpha)$.  This corresponds to the fact that in Hamming space closed ball with radius $n>l(1-1/\alpha)$ is larger than its complement.  One may also notice from Figure~\ref{fig:p-mutation-r} that the expected distance $\bE_P\{m\mid n,r\}$ may have a simple relation to the mutation radius $r=d(x,y)$.  This relation is given below.

\begin{proposition}
  The expected value and variance of the conditional probability distribution~(\ref{eq:p-mutation}) for Hamming distance $m=d(\top,y)$ of string $y\in S(x,r)$ obtained by a substitution of $d(x,y)=r$ letters in string $x\in S(\top,n)\subset\cH_\alpha^l$ are
  \begin{align}
    \bE_P\{m\mid n,r\}&=n+\left(1-\frac{n}{l(1-1/\alpha)}\right)r\,,\label{eq:e-mutation}\\
    \sigma_P^2\{m\mid n,r\}&=n\left[\frac{\alpha-2}{(\alpha-1)^2} + \frac{(l-n)(l-r)}{(1-1/\alpha)^2l(l-1)}\right]\frac{r}{l}\,. \label{eq:var-mutation}
  \end{align}
  \label{pr:mutation}
\end{proposition}

The proof uses formulae~(\ref{eq:e-lp}) and (\ref{eq:var-lp}) from Lemma~\ref{lm:moments}, where the average probabilities $\langle P_i\rangle$ and $\langle P_{ij}\rangle$ are defined as functions of parameters $n=d(\top,x)$ and $r=d(x,y)$.  See Appendix~\ref{sec:proof-e-var-mutation} for details.

Formula~(\ref{eq:e-mutation}) confirms the linear dependency of posterior expectation of distance $m=d(\top,y)$ on the mutation radius as can be seen on Figure~\ref{fig:p-mutation-r}.  The slope of this dependency is $1-n/l(1-1/\alpha)$, which is positive if $d(\top,x)=n < l(1-1/\alpha)$ (i.e. if $x$ is closer to $\top$ than the `equator') and negative if $n > l(1-1/\alpha)$.  One can see also from equation~(\ref{eq:e-mutation}) that at distance $n=l(1-1/\alpha)$ the expected value of $m=d(\top,y)$ becomes independent of the mutation radius $r=d(x,y)$ and is equal to $l(1-1/\alpha)$.  Also, the value $r=l(1-1/\alpha)$ of the mutation radius makes the expectation of distance equal to $l(1-1/\alpha)$ and independent of the parent's distance $n=d(\top,x)$, which corresponds to the uniform distribution $P_0(x)=\alpha^{-l}$ of strings.

Differentiation of equation~(\ref{eq:var-mutation}) over $r$ and setting the derivative to zero
\[
  \frac{\partial}{\partial r}\sigma^2\{m\mid n,\hat r\} = \frac{n}{l}\left[\frac{(\alpha-2)}{(\alpha-1)^2} + \frac{(l-2\hat r)(l-n)}{l(1-1/\alpha)^2(l-1)}\right]=0
\]
gives the saddle point, which is the maximum, because the second derivative is negative (observe that the radius appears only in $(l-2r)$ with the minus sign).  The mutation radius maximizing the variance of offspring's distance $d(\top,y)=m$ is
\[
  \hat r(n) = \frac{l}{2}\left[1 + \frac{(\alpha-2)(l-1)}{\alpha^2(l-n)}\right]\,.
\]
One can see that generally, unless $\alpha=2$, the maximizing mutation radius depends on the distance $d(\top,x)=n$ of the parent string, and it is not equal to the Hamming space equator $l(1-1/\alpha)=\frac{l}{2}\left[1+(\alpha-2)/\alpha\right]$.

\subsection{Maximization of  probability of beneficial mutation}

The closed-form expression~(\ref{eq:p-mutation}) combined with probability $P(r\mid n)$ of the mutation radius gives complete solution to transition probability~(\ref{eq:mutation-factorization}) between spheres around the optimum under mutation.  This makes it possible to study Markov evolution of distance distributions and solve related optimization problems.  Here we consider two simple problems that have exact solutions.

\begin{proposition}[Minimization of the expected distance after mutation]
  The expected value~(\ref{eq:e-mutation}) of Hamming distance $d(\top,y)=m$ after mutation of $x\in S(\top,n)\subset\cH_\alpha^l$ into $y\in S(\top,m)$ is minimized if the mutation radius $d(x,y)=r\in[0,l]$ has the values $r=0$ for $n<l(1-1/\alpha)$, $r=l$ for $n>l(1-1/\alpha)$, and any value for $n=l(1-1/\alpha)$.
  \label{pr:mutation-step}
\end{proposition}

\begin{proof}
  The minimization $\bE_P\{m\mid n,r\}<n$ over the mutation radius $r\in[0,l]$ follows from equation~(\ref{eq:e-mutation}).
\end{proof}

It follows that for the simple point mutation operator (Example~\ref{ex:point-mutation}) the mutation rate minimizing the expected distance after one mutation is
\[
  \hat\mu(n)=
  \begin{cases}
    0&\mbox{ if $n<l(1-1/\alpha)$}\\
    1-1/\alpha&\mbox{ if $n=l(1-1/\alpha)$}\\
    1&\mbox{ if $n>l(1-1/\alpha)$}
  \end{cases}\,.
\]
Although the offspring's expected distance $\bE\{m\mid n,r\}$ is independent of the mutation radius at distance $n=l(1-1/\alpha)$ (so that $\mu$ can have any value), we use the value $1-1/\alpha$, because several known mutation rate functions $\mu(n)$ are monotonic and pass through this value (e.g. the linear function $\mu(n)=n/l$, derived below, passes through $1-1/\alpha$ at $n=l(1-1/\alpha)$).

Note `step' mutation rate function $\hat\mu(n)$ above has the following problem.  Notice that for binary strings ($\alpha=2$) at distances $d(\top,x)=n>l/2$ the mutation rate $\mu=1$ changes the distances to $d(\top,y)=m=l-n<l/2$.  Thus, all bitstrings will be at distance $d(\top,y)\leq l/2$ after just one generation.  A similar effect will occur for strings with alphabets $\alpha>2$, but it may take several generations because of the possibility of neutral substitutions.  Therefore, after multiple generations the distribution of distances $P_{t+s}(m)=M_{\mu(n)}^s P_t(n)$, $s>1$, will not change due to mutation with rate $\mu=0$.  It is clear that the step mutation rate function is not optimal for evolution over multiple generations.  Computer simulations show that a sigmoid type mutation rate functions achieve optimality for multiple generations \cite{Belavkin11:_dyninf}, but analytic derivation of such results is not straightforward.  Another approach is to maximize the probability of mutation directly into the optimum: $x\mapsto y=\top$.

\begin{proposition}[Mutation into the optimum]
  The probability $P(m=0\mid n,r)$ that string $x\in S(\top,n)\subset\cH_\alpha^l$ mutates into string $y=\top$ is zero unless $d(x,y)=r(n)=n$, in which case the probability is
  \[
    P(m=0\mid n, r = n) = \frac{1}{|S(x,n)|}=\frac{1}{(\alpha-1)^n\binom{l}{n}}\,.
  \]
  \label{pr:mutation-optimum}
\end{proposition}

\begin{proof}
  This obvious result can be obtained formally by substituting the values $r_+=n$, $r_-=r_0=0$ into equation~(\ref{eq:p-mutation}).
\end{proof}

Substitution of the above probability, which is reciprocal of the number of elements in $S(x,n)$, into equation~(\ref{eq:mutation-factorization}) and using binomial distribution~(\ref{eq:simple-mutation}) for the mutation radius gives the following formula for the probability of transition into optimum under simple point mutation (Example~\ref{ex:point-mutation}):
\[
  P_\mu(m=0\mid n) = (\alpha-1)^{-n}\mu^n(n)[1-\mu(n)]^{l-n}\,.
\]
The optimal mutation rate $\hat\mu$ maximizing this probability is
\begin{align}
  \frac{\partial}{\partial\mu}P_\mu&=(\alpha-1)^{-n}\hat\mu^n[1-\hat\mu]^{l-n}
  \underbrace{\left(\frac{n}{\hat\mu} - \frac{l-n}{1-\hat\mu}\right)}_{=0}
  =0\qquad\implies\qquad
  \hat\mu(n)=\frac{n}{l}\,,\label{eq:linear-mutation-rate}\\
  \frac{\partial^2}{\partial\mu^2}P_\mu&=(\alpha-1)^{-n}\hat\mu^n[1-\hat\mu]^{l-n}
   \left[\underbrace{\left(\frac{n}{\hat\mu} - \frac{l-n}{1-\hat\mu}\right)^2}_{=0} - \frac{l - 2n/\hat\mu + n/\hat\mu^2}{(1-\hat\mu)^2}\right] \nonumber\\
                                   &=(\alpha-1)^{-n}\hat\mu^n[1-\hat\mu]^{l-n}
                                     \left[ \frac{l(1-l/n)}{(1-n/l)^2}\right]\leq 0\,.\nonumber
\end{align}
Indeed, one can see that for $\hat\mu=n/l$ the first derivative is zero, and the second derivative is negative, because $1-l/n\leq0$.  Observe also that as the diameter of Hamming space increases $l\to\infty$ the optimal mutation rate $\hat\mu(n)=n/l$ converges to zero for each $n\in\bN$.

These examples show that minimization of the mutation radius or mutation rate may not always be optimal in Hamming space.  The heuristic value $\mu=1/l$ \cite{Ochoa02} is optimal only at distance $d(\top,x)=1$ with respect to the criterion of Proposition~\ref{pr:mutation-optimum}.  Although the exact shapes of the mutation rate functions optimal subject to additional constraints may differ (e.g. see examples in \cite{Belavkin11:_itw11,Belavkin11:_dyninf,Belavkin11:_qbic11,Belavkin_etal16:_jomb} for constraints on the number of generations or information constraints), these functions have the common property of monotonically increasing optimal mutation rates $\mu(n)$ with distance $d(\top,\cdot)=n$ from the optimum.

\section{Recombination}
\label{sec:recombination}

\subsection{Crossover recombination in a Hamming space}

Recombination is a substitution of some letters in string $x\in\cH_\alpha^l$ by letters from another string $y\in\cH_\alpha^l$.  In this paper, we shall only consider \emph{crossover} recombination when letters $x_i$ are substituted by letters $y_i$ with the same index $i\in\{1,\ldots,l\}$.  This corresponds to Hamming metric~(\ref{eq:hamming}) accounting for differences of letters only at the same indices.

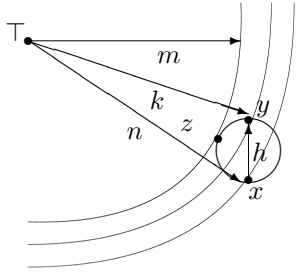
\begin{figure}[!t]
\centering
\setlength{\unitlength}{1mm}
\begin{picture}(50,40)
\linethickness{.075mm}
\put(34,16.5){\circle*{1}}
\put(34,14){$x$}
\put(34,16){\vector(0,1){8}}
\put(34,24.5){\circle*{1}}
\put(35,25.5){$y$}
\put(34.5,19){$h$}
\put(25,23){$z$}
\put(30,22){\circle*{1}}
\put(34,20.5){\circle{9.1}}
\put(5,35){\circle*{1}}
\put(2,35){$\top$}
\put(5,35){\vector(3,-2){28}}
\put(18,22){$n$}
\qbezier(5,5)(40,5)(40,40)
\put(5,35){\vector(3,-1){29}}
\put(21,26){$k$}
\qbezier(5,8)(37,8)(37,40)
\put(5,35){\vector(3,0){28}}
\put(22,32){$m$}
\qbezier(5,11)(34,10)(33,40)
\end{picture}
\caption{Recombination of string $x\in S(\top,n)$ with $y\in S(\top,k)$ into string $z\in S(\top,m)$ by crossover of $r\in[0,l]$ letters.  The number of strings in the intersection $S(\top,m)\cap I(x,y,r)$ defines the geometric probability $P(m\mid n,k,h,r)$ (\ref{eq:p-recombination-def}).}
\label{fig:recombination}
\end{figure}

Because in recombination we have to consider two parent strings $x$, $y$ and their distances from $\top$, we have a triangle $(x,y,\top)$ and three distances:
\begin{align*}
  n &= d(\top,x)\\
  k &= d(\top,y)\\
  h &= d(x,y)
\end{align*}
as shown on Figure~\ref{fig:recombination}.  The number $r\in[0,l]$ of letters that are exchanged during crossover between $x$ and $y$ will be referred to as the \emph{recombination radius}, and it can be larger than the distance $h=d(x,y)$ between two strings, because crossover may recombine identical letters.  After crossover of string $x$ with $y$ (the parents), the resulting new string $z$ (the offspring) is a mixture of $l-r$ letters from $x$ and $r$ letters from $y$, and it is created `between' its parents in the sense of Hamming distance: $d(z,y)\leq d(x,y)$ and $d(x,z)\leq d(x,y)$.  It is convenient to denote the result of recombination as $z = (1-\lambda)x \oplus \lambda y$, where $\lambda=r/l$, by analogy with convex combination in a real space, although this is only a notational convenience (hence the use of symbol $\oplus$ instead of $+$).

Recombination of $x$ with $y$ into $z$ using $r$ letters has a dual recombination $z'$, which can be formed as a substitution of the remaining $l-r$ letters from $y$ into $x$.  Thus, the dual recombination is $z'=\lambda x \oplus (1-\lambda)y$ in our `convex combination' notation.  Contrary to a real space, a mixture $z=(1-\lambda)x\oplus\lambda y$ in a Hamming space is not unique, because it depends on positions at which $r$ letters are substituted.  The totality of all possible recombinations of $r$ letters between two strings has been called a \emph{recombination potential} \cite{Eremeev08}:
\[
I(x,y,r):=\{z=(1-\lambda)x \oplus \lambda y: \lambda=r/l\}\,.
\]
The number of strings in $I(x,y,r)$ is
\[
  |I(x,y,r)| = \binom{l}{r}\,.
\]
Note that some strings in $I(x,y,r)$ may appear more than once, because different recombinations may result in the same offspring.  This makes $I(x,y,r)$ a multiset.  Also, because substitution of $r$ letters from $y$ into $x$ is the same as substitution of the remaining $l-r$ letters from $x$ into $y$, we have the following equality:
\[
I(x,y,r)=I(y,x,l-r)\,.
\]
Exchanging different letters $x_i\neq y_i$ at the same indices cannot make them equal, so that the Hamming distance between two parent strings $x$, $y$ and between their recombinations $z=(1-\lambda)x\oplus\lambda y$ and $z'=(1-\lambda)y\oplus\lambda x$ remain the same: $d(x,y)=d(z,z')$.  This implies that recombination potential $I(x,y,r)$ has a round shape: its elements belong to a sphere of diameter $h=d(x,y)$ as shown on Figure~\ref{fig:recombination}.

If distance $d(\top,\cdot)$ from the optimum $\top$ is taken as a model of (negative) fitness, then recombination is called \emph{beneficial} for parent $x\in S(\top,n)$ if it corresponds to a transition onto sphere $S(\top,m)\ni z$ of a smaller radius $m<n$, as shown on Figure~\ref{fig:recombination}.  Recombination is \emph{neutral} if $m=n$, and \emph{deleterious} if $m>n$.  Note that a recombination can be beneficial for $x\in S(\top,n)$, but not necessarily for $y\in S(\top,k)$.

\begin{example}[Crossover recombination]
Let $\top=(\mathrm{AAAAA})\in\cH_3^5$ and consider strings $x=(\mathrm{BBBAA})$ and $y=(\mathrm{BACBA})$ recombining into string $z=(\mathrm{BABBA})$, which can be obtained by a substitution of the second, fourth and the last letters in $x$ by the corresponding letters from $y$:
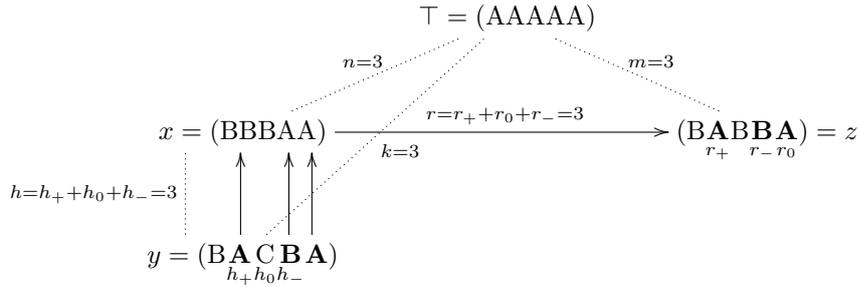
\[
\xymatrix{&\top=(\mathrm{AAAAA})\ar@{.}[ld]_{n=3}\ar@{.}[ldd]^{k=3}\ar@{.}[rd]^{m=3}&\\
  x=(\mathrm{BBBAA})\ar[rr]^{r=r_++r_0+r_-=3}&&(\mathrm{B}\underset{r_+}{\mathbf{A}}\mathrm{B}\underset{r_-}{\mathbf{B}}\underset{r_0}{\mathbf{A}})=z\\
  y=(\mathrm{B}\underset{h_+}{\mathbf{A}}\underset{h_0}{\mathrm{C}}\underset{h_-}{\mathbf{B}}\mathbf{A})\ar@{.}[u]<5ex>^{h=h_++h_0+h_-=3}\ar@<.2ex>[u]\ar@<-4ex>[u]\ar@<-6ex>[u]&&}
\]
The dual offspring is $z'=(\mathrm{BBCAA})$.  Thus, the recombination radius is $r=3$, and the recombination is neutral, because $d(\top,x)=d(\top,z)=3$.  Note that although the recombination radius was $r=3$, only two of these substitutions occurred out of $d(x,y)=h=3$ different letters; the third substitution was made for identical letters.  A substitution of the second, third and the last letters in $x=(\mathrm{BBBAA})$ by the corresponding letters from $y=(\mathrm{BACBA})$ results in string $v=(\mathrm{BACAA})$, which is closer to the optimum, $d(\top,v)=2$.
\label{ex:recombination}
\end{example}

As with mutation, we denote by $r_+$, $r_-$ and $r_0$ the numbers of beneficial, deleterious and neutral substitutions respectively.  They satisfy the same equations as~(\ref{eq:rm-sum}) and (\ref{eq:rm-difference}):
\begin{align}
  &r_++r_-+r_0=r\,, \label{eq:rr-sum}\\
  &r_+-r_-=n-m\,. \label{eq:rr-difference}
\end{align}
Here, $n=d(\top,x)$ and $m=d(\top,z)$.  The numbers $r_+$ and $r_-$ count beneficial or deleterious substitutions for parent $x$, but not necessarily for parent $y$.  These substitutions can only occur for $h=d(x,y)$ different letters.  Therefore, $r_+\leq h_+$ and $r_-\leq h_-$, where $h_+$ and $h_-$ are the maximum possible numbers of beneficial and deleterious substitutions for $x$ and $h_++h_-\leq h=d(x,y)$.  Denoting by $h_0$ the maximum number of possible neutral substitutions out of $h$ different letters, the equations for these maximal numbers are
\begin{align}
  &h_++h_-+h_0=h\,, \label{eq:nr-sum}\\
  &h_+-h_-=n-k\,. \label{eq:nr-difference}
\end{align}
These equations are identical to~(\ref{eq:rm-sum}) and (\ref{eq:rm-difference}) if parent $y$ is considered as a mutated version of parent $x$ with $h=d(x,y)$ considered as the mutation radius.  As mentioned earlier, substitutions during crossover may occur also among the $l-d(x,y)=l-h$ identical letters (as in Example~\ref{ex:recombination}).  Such substitutions are neutral, and therefore it is possible that $r_0>h_0$.  In fact, the range of recombination radius is $r\in[0,l]$, and it can be larger than $h=d(x,y)$.  Changes between the parent and offspring strings are defined only by the distance $h=d(x,y)$, and we shall refer to it as \emph{recombination capacity}.

The range of recombination capacity is defined by the triangle inequalities:
\[
  |n-k|\leq h\leq n+k
\]
As with mutation radius, there are no potential neutral substitutions ($h_0=0$) at the extreme values $h=|n-k|$ or $h=n+k$.

\begin{proposition}[Dual recombination]
  If crossover recombination of $x\in S(\top,n)$ with $y\in S(\top,k)$ by exchanging $r\in[0,l]$ letters results in string $z\in S(\top,m)$, then the dual recombination $z'\in S(\top,m')$ is at distance $d(\top,z')=m'$:
  \[
    m'=n+k-m\,.
  \]
  \label{pr:dual}
\end{proposition}

See Appendix~\ref{sec:pr:dual} for the proof.  Intuitively, if $x$ receives $n-m$ letters $y_i=\top_i$ from $y$, then the Hamming distance reduces from $d(\top,x)=n$ to $d(\top,z)=n-(n-m)=m$.  At the same time, string $y$ receives $n-m$ letters $x_i\neq\top_i$ from $x$ (the dual recombination), which means that $d(\top,y)=k$ changes to $d(\top,z')=k+n-m$.

\subsection{Evolutionary dynamics under recombination}

Crossover recombination can be viewed as a transition from the pair $(x,y)\in\cH_\alpha^l\times\cH_\alpha^l$ of two parent strings to the pair $(z,z')\in\cH_\alpha^l\times\cH_\alpha^l$ of recombination $z$ and its dual $z'$.   The corresponding transition probability $P(z,z'\mid x,y)$ induces a transformation of the distribution $P_t(x,y)$ of the parent pairs into the distribution $P_{t+1}(z,z')$ of the offspring pairs:
\[
  P_{t+1}(z,z')=\sum_{(x,y)\in\cH_\alpha^l\times\cH_\alpha^l}P(z,z'\mid x,y)\,P_t(x,y)\,.
\]
Joint distributions $P_t(x,y)$ represent pairing probabilities of the parent strings for recombination, and they may depend on various properties such as fitness of individuals as well as their similarity.  Thus, generally $P_t(x,y)=P_t(y\mid x)P_t(x)$ and $P_t(y\mid x)\neq P_t(y)$.

The joint pairing distributions $P_t(x,y)$ induce the corresponding joint distributions of distances $n=d(\top,x)$ and $k=d(\top,y)$ of the paired strings:
\[
  P_t(n,k):=P_t\{x\in S(\top,n),y\in S(\top,k)\}\,.
\]
These joint distributions can be formed as products $P_t(k\mid n)P_t(n)$, where $P_t(n):=P_t\{x\in S(\top,n)\}$.  If strings $(x,y)$ are paired independently, then $P_t(k,n)=P_t(k)P_t(n)$.  However, generally $P_t(k\mid n)\neq P_t(k)$.

\begin{example}[Matching]
  If string $x\in S(\top,n)$ is paired with string $y\in S(\top,k)$ at equal distances $d(\top,x)=n=k=d(\top,y)$, then
  \[
    P(k\mid n)=\delta_n(k)\,,\qquad
    P_t(n,k)=
    \begin{cases}
      P_t(n) & \text{if $k=n$}\\
      0 & \text{otherwise}
    \end{cases}\,.
  \]
  This joint distribution may occur as an equilibrium solution to a minimax problem \cite{Neumann-Morgenstern}, when both parents minimize distances $d(\top,\cdot)$ of the strings they are recombined with (i.e. maximizing fitness of their partners).
  \label{ex:matching}
\end{example}

Apart from the distances from $\top$, each pair of strings $x,y\in\cH_\alpha^l$ is characterized also by their distance $h=d(x,y)$ (recombination capacity).  The joint distribution $P_t(n,k,h)=P_t(h\mid n,k)P_t(n,k)$ is also defined by the pairing distribution $P_t(x,y)$.

\begin{example}[Random pairing from a sphere]
  Consider string $x\in S(\top,n)$ paired with string $y\in S(\top,k)$ (i.e. distances $d(\top,x)=n$ and $d(\top,y)=k$ are fixed).  If string $y\in S(\top,k)$ is chosen uniformly at random from $S(\top,k)$, then conditional probability $P(h\mid n,k)$ is defined by the intersection of spheres $S(x,h)$ and $S(\top,k)$ (see Figure~\ref{fig:recombination}):
  \begin{align*}
    P(h\mid n,k)&=\frac{|S(x,h)\cap S(\top,k)|_{d(\top,x)=n}}{|S(\top,k)|}\\
                &=\frac{\sum\limits_{h_+=0}^h(\alpha-2)^{h_0}\binom{n-h_+}{h_0}(\alpha-1)^{h_-}\binom{l-n}{h_-}\binom{n}{h_+}}{(\alpha-1)^k\binom{l}{k}}\,.
  \end{align*}
  The latter formula is obtained using equation~(\ref{eq:intersection-sphere}) for intersection of spheres with distance $d(x,y)=h=h_++h_-+h_0$ treated as the mutation radius $r=r_++r_-+r_0$ and substituting $k$ for $m$.  The probability is zero if the triangle inequalities $|n-k|\leq h\leq n+k$ are not satisfied.
  \label{ex:pairing-h-random}
\end{example}

\begin{example}[Pairing at specific distance]
  One can try to pair strings choosing a specific value of the distance $d(x,y)=h$ between the parent strings (assuming that the current population has individuals satisfying this equality constraint).  The range of $h\in[0,l]$ is defined by the triangle inequalities: $|n-k|\leq h\leq n+k$ (and $n+k\leq l$).  For example, choosing the maximum value $h=\min\{n+k,l\}$ corresponds to the probability
  \[
    P(h\mid n,k)=\delta_{\min\{n+k,l\}}(h)\,.
  \]
  Minimization of $h$ corresponds to $\delta_{|n-k|}(h)$, and the average $h=\frac{1}{2}(|n-k|+n+k)=\max\{n,k\}$ to $\delta_{\max\{n,k\}}(h)$.
  \label{ex:pairing-h-specific}
\end{example}

Recombination of strings $(x,y)$ into $(z,z')$, where $z'$ denotes the dual recombination, results in a transformation of distances $n=d(\top,x)$ and $k=d(\top,y)$ into distances $m=d(\top,z)$ and $m'=d(\top,z')$.  The joint distributions of distance pairs $P_t(n,k):=P_t\{x\in S(\top,n),y\in S(\top,k)\}$ and $P_{t+1}(m,m'):=P_{t+1}\{z\in S(\top,m),z'\in S(\top,m')\}$ are transformed as follows:
\[
  P_{t+1}(m,m')=\sum_{n=0}^l\sum_{k=0}^l P(m,m'\mid n,k)\,P_t(n,k)\,,
\]
where $P(m,m'\mid n,k)$ is the transition probability between the pairs of spheres of radii $(n,k)$ and $(m,m')$:
\[
  P(m,m'\mid n,k):=P\{z\in S(\top,m),z'\in S(\top,m')\mid x\in S(\top,n),y\in S(\top,k)\}\,.
\]
The analysis is similar to mutation, but the transformations are now applied to joint distributions of distance pairs.  One can also obtain the transformation of distance distribution $P_t(n):=P_t\{x\in S(\top,n)\}$ into $P_{t+1}(m):=P_{t+1}\{z\in S(\top,m)\}$:
\[
  P_{t+1}(m)=\sum_{m'=0}^l\sum_{n=0}^l\sum_{k=0}^lP(m,m'\mid n,k)\,P(k\mid n)P_t(n)\,.
\]
Note that $m'=n+k-m$ for crossover recombination (Proposition~\ref{pr:dual}), so that
\[
  P(m,m'\mid n,k)=
  \begin{cases}
    P(m\mid n,k) & \text{if $m'=n+k-m$}\\
    0 & \text{otherwise}
  \end{cases}\,.
\]
Thus, the summation over $m'\in[0,l]$ above is not necessary, and it is sufficient to derive the expressions using only probability $P(m\mid n,k)$.

If the transition kernels $P(m\mid n,k)$ and $P(k\mid n)$ are time invariant, then the linear operator
\begin{equation}
  R(\cdot)=\sum_{n=0}^l\left[\sum_{k=0}^lP(m\mid n,k)\,P(k\mid n)\right]\,(\cdot)
  \label{eq:operator-recombination}
\end{equation}
acting on distributions $P_t(n)$ of distances $d(\top,x)=n\in[0,l]$ generates the entire evolution $\{P_t\}_{t\geq0}$ due to recombination as $P_{t+s}=R^sP_t$.  This can be used in simulations to analyse the effects of recombination and pairing strategies on evolution.

In Section~\ref{sec:mutation} on mutation, we expanded transition probability $P(m\mid n)$ over all values of the mutation radius $r\in[0,l]$ (\ref{eq:mutation-factorization}).  Similarly, here we expand the transition probability $P(m\mid n,k)$ over all values of the recombination radius $r\in[0,l]$ and recombination capacity $h=d(x,y)$:
\begin{equation}
  P(m\mid n,k)=\sum\limits_{r=0}^l\sum\limits_{h=0}^l P(m\mid n,k,h,r)\,\underbrace{P(r\mid n,k,h)}_{\text{Recombination}}\,\underbrace{P(h\mid n,k)}_{\text{Pairing}}\,.
  \label{eq:recombination-factorization}
\end{equation}
The probability $P(h\mid n,k)$ has been discussed in Examples~\ref{ex:pairing-h-random} and \ref{ex:pairing-h-specific}.  The probability $P(r\mid n,k,h)$ of recombination radius $r\in[0,l]$ can be determined from the analysis of the recombination operator.

\begin{example}[Uniform crossover]
In this form of recombination, letters $x_i$ and $y_i$ at each position $i\in\{1,\ldots,l\}$ in the parent strings are swapped with probability $\nu\in[0,1]$, called the \emph{recombination rate}, independently of letters $x_j$ and $y_j$ at other positions.  In this case, $P(r\mid n,k,h)$ is the binomial distribution:
\[
  P_\nu(r\mid n,k,h)=\binom{l}{r}\nu^r(n,k,h)[1-\nu(n,k,h)]^{l-r}\,.
\]
The rate $\nu$ may be different for different values of $n,k,h\in[0,l]$, so that the recombination operator depends on the recombination rate control function $\nu(n,k,h)$.
\label{ex:uniform-crossover}
\end{example}

\begin{example}[One point crossover]
In this form of recombination a single index $i\in\{1,\ldots,l\}$ is selected in the parent strings $x$ and $y$, and all letters $x_j$, $y_j$ with $j\geq i$ are swapped.  Thus, if $i=\lfloor l/2\rceil$, then approximately half of the letters are swapped, and the recombination radius is $r=\lfloor l/2\rceil$ (here $\lfloor\cdot\rceil$ denotes the nearest integer).  In this case, $P(r\mid n,k,h)$ is the Dirac distribution:
\[
  P(r\mid n,k,h)=\delta_{\lfloor l/2\rceil}(r)\,.
\]
Observe that the recombination radius in this case is equal to $r=l-i$, where $i$ is the index of one point crossover.  Potentially, one can define one-point crossover with variable index $i=l-r$, where the value $r=r(n,k,h)$ of the recombination radius may depend on distances $n$, $k$ and $h$.
\label{ex:one-point-crossover}
\end{example}

Using factorization~(\ref{eq:recombination-factorization}), the recombination operator~(\ref{eq:operator-recombination}) acting on distributions of distances from $\top$ takes the form
\[
  R(\cdot)=\sum_{n=0}^l\left[\sum_{k=0}^l\sum_{r=0}^l\sum_{h=0}^l P(m\mid n,k,h,r)\,P(r\mid n,k,h)\,P(h\mid n,k)P(k\mid n)\right]\,(\cdot)\,,
\]
where probabilities $P(k,h\mid n)=P(h\mid n,k)P(k\mid n)$ are defined by the pairing strategy (Examples~\ref{ex:matching}, \ref{ex:pairing-h-random}, \ref{ex:pairing-h-specific}) and $P(r\mid n,k,h)$ by the crossover process (Examples~\ref{ex:uniform-crossover}, \ref{ex:one-point-crossover}).  The unknown conditional probability $P(m\mid n,k,h,r)$ will be determined in the next section.

\subsection{Geometric probability of recombination onto a sphere}

Similar to geometric probability $P(m\mid n,r)$ defined in~(\ref{eq:p-mutation-def}) for mutation, the probability $P(m\mid n,k,h,r)$ in factorization~(\ref{eq:recombination-factorization}) represents a purely geometric problem depicted on Figure~\ref{fig:recombination}:
\begin{multline}
  P(m\mid n,k,h,r):=\\
  P\{z\in S(\top,m)\mid x\in S(\top,n),y\in S(\top,k)\cap S(x,h), r\in[0,l]\}
  \label{eq:p-recombination-def}
\end{multline}
Solution to this geometric problem requires counting the number of elements in certain subsets.  The sought offspring strings are in the intersection of recombination potential $I(x,y,r)$ and sphere $S(\top,m)$.  The difference $n-m$ of the radii of spheres $S(\top,n)$ and $S(\top,m)$ defines the difference $r_+-r_-$ of beneficial and deleterious substitutions into string $x$ from $y$ (\ref{eq:rr-difference}).  The upper bounds $r_+\leq h_+$ and $r_-\leq h_-$ can be defined from the conditions on the parent strings: $x\in S(\top,n)$ and $y\in S(\top,k)\cap S(x,h)$ (i.e. by the distances $n$, $k$ and $h$ for the triangle $(\top,x,y)$).  Indeed, $h_+\leq h=d(x,y)$ and $h_-=h_+-(n-k)$ by equation~(\ref{eq:nr-difference}), where $n=d(\top,x)$ and $k=d(\top,y)$.  It is convenient to count strings in the intersection $S(\top,m)\cap I(x,y,r)$ by grouping them based on the values $h_+\in[0,h]$ of possible beneficial substitutions.  We shall denote by $[h_+]$ the class of all strings $y$ that have $r_+\leq h_+$ beneficial substitutions into $x$.

\begin{lemma}[Intersection of sphere and recombination potential]
  Let $x\in S(\top,n)\subset\cH_\alpha^l$, and let $[h_+]$ be the class of all strings $y\in S(\top,k)\cap S(x,h)$ such that the number of beneficial substitutions from $y$ into $x$ (reducing the distance $d(\top,x)=n$) be at most $h_+\leq h$.  Then the number of elements in the intersection of $S(\top,m)$ with recombination potential $I(x,y,r)$ for $y\in[h_+]$ is
\begin{equation}
  \Bigl|S(\top,m)\cap I(x,y,r)\Bigr|_{y\in[h_+]}=\sum\limits_{r_+=0}^{h_+}\binom{l-h_+-h_-}{r-r_+-r_-}\binom{h_-}{r_-}\binom{h_+}{r_+}\,,
  \label{eq:intersection-interval-h+}
\end{equation}
where $h_-=h_+-(n-k)\geq0$, $r_-=r_+-(n-m)\geq0$ and $r_+\in[0,h_+]$, $r-r_+-r_-\geq0$.
\label{lm:recombination-h+}
\end{lemma}

\begin{proof}
  Given the maximum numbers $h_+$ and $h_-=h_+-(n-k)$ of beneficial and deleterious substitutions (observe that $h_-$ depends on $n=d(\top,x)$ and $k=d(\top,y)$), the remaining $l-h_+-h_-$ letters can only make neutral substitutions.  Thus, for specific values of $r_+\leq h_+$, $r_-\leq h_-$, the total number of combinations is
  \[
    \binom{l-h_+-h_-}{r-r_+-r_-}\binom{h_-}{r_-}\binom{h_+}{r_+}\,,
  \]
  where $r_-=r_+-(n-m)$ from equation~(\ref{eq:rr-difference}), and $r-r_+-r_-\geq0$ is the total number of neutral recombinations from equation~(\ref{eq:rr-sum}).  Feasible values of $r_+\in[0,h_+]$ are determined from non-negativity of $r_-$ and $r-r_+-r_-$.  The total number is obtained by summing over all feasible values of $r_+\in[0,h_+]$.
\end{proof}

The maximum number of beneficial substitutions is bounded above $h_+\leq h=d(x,y)$, and adding the numbers in equation~(\ref{eq:intersection-interval-h+}) for all $h_+\in[0,h]$ would account for all strings in the intersection $S(\top,m)\cap I(x,y,r)$.  However, the classes $[h_+]$ are not distributed uniformly as they have different sizes.  Indeed, the totality of all strings $y$ to be recombined with $x\in S(\top,n)$ is the intersection of spheres $S(\top,k)$ and $S(x,h)$, as shown on Figure~\ref{fig:recombination}.  The number of strings in this intersection is given by formula~(\ref{eq:intersection-sphere}) in Lemma~\ref{lm:spheres}, where instead of mutation radius $r=r_0+r_-+r_0$ we have to use recombination capacity $h=h_++h_-+h_0$ and substituting $k$ for $m$.  In fact, the number of strings $y$ in the intersection $S(\top,k)\cap S(x,h)$ with specific values of the maximum numbers $h_+$, $h_-$ and $h_0$ is
\[
(\alpha-2)^{h_0}\binom{n-h_+}{h_0}(\alpha-1)^{h_-}\binom{l-n}{h_-}\binom{n}{h_+}\,.
\]
This can be used to derive probability~(\ref{eq:p-recombination-def}).

\begin{theorem}[Geometric probability of recombination onto a sphere]
  The probability $P(m\mid n,k,h,r)$ that crossover recombination of $r\in[0,l]$ letters in string $x\in S(\top,n)\subset\cH_\alpha^l$ with string $y\in S(\top,k)\cap S(x,h)$ results in string $z\in S(\top,m)$ is
  \begin{multline}
    P(m\mid n,k,h,r)=\\
    \frac{\sum\limits_{h_+=0}^h(\alpha-2)^{h_0}\binom{n-h_+}{h_0}(\alpha-1)^{h_-}\binom{l-n}{h_-}\binom{n}{h_+} \sum\limits_{r_+=0}^{h_+}\binom{l-h_+ - h_-}{r-r_+-r_-}\binom{h_-}{r_-}\binom{h_+}{r_+}}{\binom{l}{r}\sum\limits_{h_+=0}^h(\alpha-2)^{h_0}\binom{n-h_+}{h_0}(\alpha-1)^{h_-}\binom{l-n}{h_-}\binom{n}{h_+}}
    \label{eq:p-recombination}
  \end{multline}
  with $r_+\in[0,h_+]$, $h_+\in[0,h]$ and the numbers $h_-\geq0$, $h_0\geq0$, $r_-\geq0$, $r-r_+-r_-\geq0$ defined by the equations
  \[
    h_-=h_+-(n-k)\,,\qquad h_0=h-2h_++(n-k)\,,\qquad r_-=r_+-(n-m)\,.
  \]
  The probability is zero if the triangle inequalities $|n-k|\leq h\leq n+k$ or $|n-m|\leq \min\{r,h\}\leq n+m$ are not satisfied.
  \label{th:recombination}
\end{theorem}

\begin{proof}
  The probability $P(m\mid n,k,h,r)$ can be expressed as the following sum of products of conditional probabilities $P(m\mid n,k,h,r,h_+)$ and $P(h_+\mid n,h,k)$ for all values of $h_+\in[0,h]$:
  \begin{equation}
    P(m\mid n,k,h,r)=\sum\limits_{h_+=0}^h P(m \mid n,k,h,r,h_+)P(h_+\mid n,k,h)\,.
    \label{eq:p-recombination-marginal}
  \end{equation}
  Observe that the recombination radius $r$ does not occur in probability $P(h_+\mid n,h,k)$.  This is because the maximum number of possible beneficial recombinations $h_+\in[0,h]$ is defined solely by the parent strings, and it is independent of $r$.

  Here, $P(m\mid n,k,h,r,h_+)$ is the probability that an offspring $z$ in the recombination potential $I(x,y,r)$ is at distance $m=d(\top,z)$ from the optimum subject to the condition that the parent string $y$ has at most $h_+$ potential beneficial recombinations (fixing specific three points $(\top,x,y)$ in a Hamming space also fixes their distances $n=d(\top,x)$, $k=d(\top,y)$ and $h=d(x,y)$).  This probability is the ratio of strings in the intersection of sphere $S(\top,m)$ with potential $I(x,y,r)$ with the condition $y\in[h_+]$ over all offspring strings in $I(x,y,r)$:
  \[
    P(m\mid n,k,h,r,h_+)=\frac{|S(\top,m)\cap I(x,y,r)|_{y\in[h_+]}}{|I(x,y,r)|}\,.
  \]
  The number $|S(\top,m)\cap I(x,y,r)|_{y\in[h_+]}$ is given by equation~(\ref{eq:intersection-interval-h+}), and the number of elements in the potential $I(x,y,r)$ is $\binom{l}{r}$, so that
  \begin{equation}
    P(m\mid n,k,h,r,h_+)=\frac{\sum\limits_{r_+=0}^{h_+}\binom{l-h_+-h_-}{r-r_+-r_-}\binom{h_-}{r_-}\binom{h_+}{r_+}}{\binom{l}{r}}\,,
    \label{eq:p-intersection-h+}
  \end{equation}
  where $r_+-r_-=n-m$ with the constraints $r_-\geq0$, $r-r_+-r_-\geq0$ and $h_-=h_+-(n-k)\geq0$.
  
  Probability $P(h_+\mid n,k,h)$ is the ratio of parent strings $y$ with at most $h_+$ beneficial recombinations in the intersection $S(\top,k)\cap S(x,h)$ out of all parent strings in this intersection:
  \[
    P(h_+\mid n,k,h)=\frac{|S(\top,k)\cap S(x,h)|_{d(\top,x)=n,\,d(\top,y)=k,\,y\in[h_+]}}{|S(\top,k)\cap S(x,h)|_{d(\top,x)=n}}\,.
  \]
  Using equation~(\ref{eq:intersection-sphere}) for intersection of the spheres and making the described earlier substitutions this probability is
  \begin{equation}
    P(h_+\mid n,k,h)=\frac{(\alpha-2)^{h_0}\binom{n-h_+}{h_0}(\alpha-1)^{h_-}\binom{l-n}{h_-}\binom{n}{h_+}}{\sum\limits_{h_+=0}^h(\alpha-2)^{h_0}\binom{n-h_+}{h_0}(\alpha-1)^{h_-}\binom{l-n}{h_-}\binom{n}{h_+}}\,,
    \label{eq:p-h+}
  \end{equation}
  where $h_+-h_-=n-k$ and $h_++h_-+h_0=h$ with the constraints $h_-\geq0$, $h_0\geq0$.  The final formula~(\ref{eq:p-recombination}) is obtained by the substitution of equations~(\ref{eq:p-intersection-h+}) and (\ref{eq:p-h+}) into (\ref{eq:p-recombination-marginal}).
\end{proof}

\begin{remark}
  The summations in equation~(\ref{eq:p-recombination}) are shown across all $r_+\in[0,h_+]$ and $h_+\in[0,h]$, but it is important to check also that all other indices are non-negative: $r_-\geq0$, $r-r_+-r_-\geq0$, $h_-\geq0$ and $h_0\geq0$.  The triangle inequalities $|n-k|\leq h\leq n+k$ imply the following bounds $\max\{0,n-k\} \leq h_+\leq\frac12(h+n-k)\leq h$ (and similar for $r_+$), which can be used for a more efficient implementation.
  \label{rem:range-n}
\end{remark}

\begin{figure}[!th]
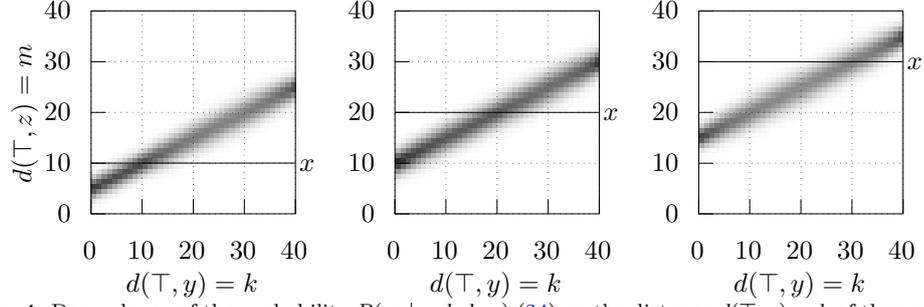

  \centering
  \input{p-recombine-k-10-40}
  \input{p-recombine-k-20-40}
  \input{p-recombine-k-30-40}
  \caption{Dependency of the probability $P(m\mid n, k, h, r)$ (\ref{eq:p-recombination}) on the distance $d(\top,y)=k$ of the second string (abscissae) in space $\cH_4^{40}$.  Ordinates show the resulting distance $d(\top,z)=m$ after crossover.  Three charts correspond to three distances $d(\top,x)=n\in\{10,20,30\}$ of the first parent string.  Grayscale represents probability $P\in[0,1]$.  Parameters $d(x,y)=h=\max\{n,k\}$ and $r=20$.}
  \label{fig:p-recombination-k}
\end{figure}

\begin{figure}[!ht]
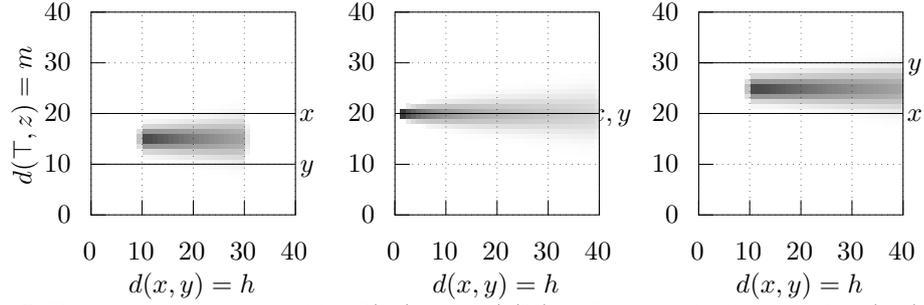

  \centering
  \input{p-recombine-h-kn-10-20-40}
  \input{p-recombine-h-kn+0-20-40}
  \input{p-recombine-h-kn+10-20-40}
  \caption{Dependency of the probability $P(m\mid n,k,h,r)$ (\ref{eq:p-recombination}) on the recombination capacity $d(x,y)=h$ (abscissae) in space $\cH_4^{40}$.  Ordinates show the resulting distance $d(\top,z)=m$ after crossover.  Three charts correspond to three distances $d(\top,y)=k\in\{10,20,30\}$. Parameters $d(\top,x)=n=20$ and $r=20$.  Grayscale represents probability $P\in[0,1]$.}
  \label{fig:p-recombination-h}
\end{figure}

\begin{figure}[!ht]
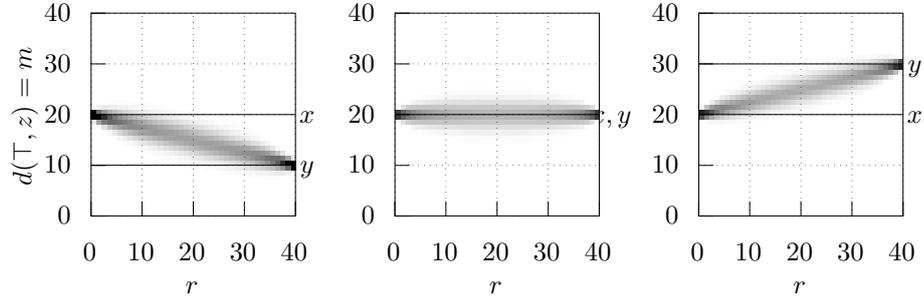

  \centering
  \input{p-recombine-r-kn-10-20-40}
  \input{p-recombine-r-kn+0-20-40}
  \input{p-recombine-r-kn+10-20-40}
  \caption{Dependency of the probability $P(m\mid n,k,h,r)$ (\ref{eq:p-recombination}) on the recombination radius $r\in[0,l]$ (abscissae) in space $\cH_4^{40}$.  Ordinates show the resulting distance $d(\top,z)=m$ after crossover.  Three charts are shown for three values of distance $d(\top,y)=k\in\{10,20,30\}$.  Parameters $d(\top,x)=n=20$ and $d(x,y)=h=\max\{n,k\}$.  Grayscale represents probability $P\in[0,1]$.}
  \label{fig:p-recombination-r}
\end{figure}

\begin{example}[Binary case $\alpha=2$]
  In the binary case there are no neutral substitutions among $h=d(x,y)$ different letters, and therefore $h_0=0$, $h=h_++h_-$ and $h_+-h_-=n-k$ define one possible value $h_+=\frac12(h+n-k)$.  Probability~(\ref{eq:p-h+}) becomes
  \[
    P(h_+\mid n,k,h)=\delta_{\frac12(h+n-k)}(h_+)\,.
  \]
  Substituting $h_+=\frac12(h+n-k)$ and using the conditions $r_+\in[0,h_+]$, $r_-=r_+-(n-m)\geq0$, $r_0=r-2r_++(n-m)\geq0$ formula~(\ref{eq:p-recombination}) reduces to
  \[
    P(m\mid n,k,h,r)=\frac{\sum\limits_{r_+=0}^{h_+}\binom{l-h}{r-2r_++(n-m)}\binom{h-h_+}{r_+-(n-m)}\binom{h_+}{r_+}}{\binom{l}{r}}\,.
  \]
  This formula is also valid for $\alpha>2$ when the distance $d(x,y)=h$ is minimized ($h=|n-k|$) or maximized ($h=n+k$), because there are $h_0=0$ possible neutral substitutions among $h=d(x,y)$ different letters in these cases.
  \label{ex:recombination-2}
\end{example}

Conditional probability~(\ref{eq:p-recombination}) was implemented in a digital computer using Common Lisp programming language, and Figures~\ref{fig:p-recombination-k}--\ref{fig:p-recombination-r} illustrate its dependency on parameters $n$, $k$, $h$ and $r$ in Hamming space $\cH_4^{40}$ ($\alpha=4$, $l=40$).  Ordinates on all charts show the resulting distance $m=d(\top,z)$ of offspring after crossover.  The grayscale represents different values of probability $P(m\mid n,k,h,r)$ with white corresponding to $P=0$ and black to $P=1$.

Figure~\ref{fig:p-recombination-k} shows the effect of distance $k=d(\top,y)$ of the second parent (abscissae) relative to the distance $n=d(\top,x)$ of the first parent shown on three charts for $n\in\{10,20,30\}$.  On all charts recombination radius was $r=\lfloor l/2\rceil=20$ and capacity $d(x,y)=h=\max\{n,k\}$.  One can see that the resulting distribution appears to have linear dependency on distance $d(\top,y)=k$, and (as expected) crossover with $d(\top,y)<d(\top,x)$ increases the chance of beneficial recombination.

Figure~\ref{fig:p-recombination-h} shows the effect of recombination capacity $h=d(x,y)$ (abscissae).  Three charts correspond to distances $d(\top,y)=k\in\{10,20,30\}$ and $d(\top,x)=n=20$.  Recombination radius was $r=\lfloor l/2\rceil=20$ on all charts.  One can see that recombination capacity increases the variance of the resulting distribution of $d(\top,z)=m$ with the maximum variance achieved for $h=n+k$.

Figure~\ref{fig:p-recombination-r} shows the effect of recombination radius $r$ (abscissae).  Three charts correspond to three distances $d(\top,y)=k\in\{10,20,30\}$ of the second parent and distance $d(\top,x)=n=20$.   Recombination capacity was $d(x,y)=h=\max\{n,k\}$ on all charts.  One can see that small recombination radius concentrates the probability at distance of the first parent, while increasing the radius concentrates the probability at distance of the second parent (as expected).  One can see also the variance of the offspring's distance $d(\top,z)=m$ appears to be maximized at $r=\lfloor l/2 \rceil=20$.

The above observations about the expected value and variance of distance $d(\top,z)=m$ after crossover are confirmed by the corresponding formulae below.

\begin{proposition}
  The expected value and variance of conditional probability distribution~(\ref{eq:p-recombination}) for Hamming distance $m=d(\top,z)$ of string $z$ obtained by a crossover recombination of $r\in[0,l]$ letters in string $x\in S(\top,n)$ from string $y\in S(\top,k)\cap S(x,h)$ in a Hamming space $\cH_\alpha^l$ are
  \begin{align}
    \bE_P\{m\mid n,k,h,r\}&=n+\frac{(k-n)}{l}r\,,\label{eq:e-recombination}\\
    \sigma_P^2\{m\mid n,k,h,r\}&=\left[h-\langle h_0\rangle - \frac{(n-k)^2}{l}\right]\frac{r(l-r)}{l(l-1)}\,.\label{eq:var-recombination}
  \end{align}
  Here $\langle h_0\rangle :=\bE\{h_0\mid n,k,h\}$ is the expected maximum number of neutral substitutions among $h=d(x,y)$ different letters, which is computed using formula~(\ref{eq:p-h+}) for conditional distribution $P(h_+\mid n,k,h)$ and using the relation $h_0 = h-2h_++n-k$.
  \label{pr:recombination}
\end{proposition}

The proof uses formulae~(\ref{eq:e-lp}) and (\ref{eq:var-lp}) from Lemma~\ref{lm:moments}, where probabilities $\langle P_i\rangle$ and $\langle P_{ij}\rangle$ are defined as functions of parameters $d(\top,x)=n$, $d(\top,y)=k$, $d(x,y)=h$ and recombination radius $r\in[0,l]$.  See Appendix~\ref{sec:proof-e-var-recombination} for details.

Formula~(\ref{eq:e-recombination}) confirms `linear' dependency that can be seen on Figure~\ref{fig:p-recombination-r}.  The slope of this dependency is defined by the difference $d(\top,y)-d(\top,x)=k-n$ of distances of the two parent strings.  Similarly, the variance depends only on the squared difference $(n-k)^2$ in equation~(\ref{eq:var-recombination}).  Thus, the slope and the variance of distance distribution after crossover are invariant under translation $n\mapsto n+m$ and $k\mapsto k+m$ of distances from the optimum (since $n+m-k-m=n-k$).

Formula~(\ref{eq:var-recombination}) shows also that the variance is maximized at equal distances $d(\top,x)=d(\top,y)$ (i.e. $(n-k)^2=0$) and when recombination capacity $d(x,y)=h$ is maximized: $h=n+k$ (in which case $\langle h_0\rangle=0$).  These effects of distances $d(\top,x)=n$, $d(\top,y)=k$ and capacity $d(x,y)=h$ on the variance are shown on Figure~\ref{fig:p-recombination-h-minmax}.

\begin{figure}[!t]
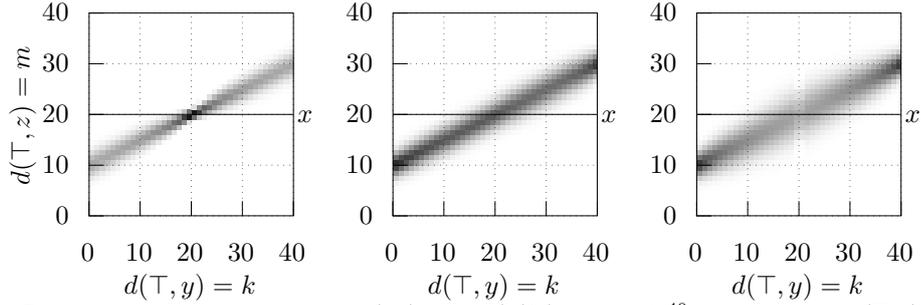

  \centering
  \input{p-recombine-k-minh-20-40}
  \input{p-recombine-k-20-40}
  \input{p-recombine-k-maxh-20-40}
  \caption{Dependency of the probability $P(m\mid n, k, h, r)$ (\ref{eq:p-recombination}) in space $\cH_4^{40}$ on the distance $d(\top,y)=k$ of the second string (abscissae) and using three strategies for recombination capacity $d(x,y)=h$: minimization $h=|n-k|$ (left), mean $h=\max\{n,k\}$ (centre), maximization $h=n+k\leq l$ (right).  Ordinates show the resulting distance $d(\top,z)=m$ after crossover.  Recombination radius $r=20$.  Grayscale represents probability $P\in[0,1]$.}
  \label{fig:p-recombination-h-minmax}
\end{figure}

The dependency of variance~(\ref{eq:var-recombination}) on the recombination radius $r\in[0,l]$ is particularly interesting, as this parameter is independent of the others.  Maximization gives the following result:
\[
  \frac{\partial}{\partial r}\sigma^2\{m\mid n,k,h,r\}=\left[h-\langle h_0\rangle - \frac{(n-k)^2}{l}\right]\frac{l-2\hat r}{l(l-1)}=0\qquad\implies\qquad
  \hat r=\frac{l}{2}\,,
\]
(the second derivative is negative).  Therefore, the variance of distances after crossover recombination is maximized when exactly half of the letters in the strings are recombined.  This corresponds to the one-point crossover at index $i=\lfloor l/2\rceil$.  For uniform crossover with rate $\nu=1/2$ recombination radius is random with the mean value $l/2$.

Finally, let us show and discuss the following symmetry property of beneficial and deleterious crossover recombinations.

\begin{proposition}
  Geometric probability~(\ref{eq:p-recombination}) has the following symmetry:
  \[
    P(m\mid n,k,h,r)=P(n+k-m\mid n,k,h,l-r)\,.
  \]
  \label{pr:symmetry}
\end{proposition}

See Appendix~\ref{sec:proof-symmetry} for the proof.  Notice that for strings at equal distances $d(\top,x)=d(\top,y)=n$ the probability that recombinations is beneficial $m=n-(n-m)<n$ is equal to the probability that the dual recombination is deleterious $m'=n+n-m>n$.  Thus, chances of beneficial and deleterious recombinations are in a certain sense equal, and this property is uniform across the entire Hamming space $\cH_\alpha^l$.  This is different from mutation, because beneficial mutations are less frequent than deleterious mutations for all strings with $d(\top,x)<l(1-1/\alpha)$.

\subsection{Maximization of probability of beneficial recombination}

The closed-form expression~(\ref{eq:p-recombination}) combined with probabilities of the recombination radius $P(r\mid n,k,h)$ and recombination capacity $P(h\mid n,k)$ gives complete solution to transition probability~(\ref{eq:recombination-factorization}) between pairs of spheres around optimum after crossover.  This makes it possible to maximize the probability of beneficial crossover recombination.  As with mutation, however, this optimization problem can be defined in many different ways (e.g. subject to a constraint on the number of generations).  Below we consider two simplified problems that have exact solutions.

\begin{proposition}[Minimization of the expected distance after crossover]
  The expected value~(\ref{eq:e-recombination}) of Hamming distance $d(\top,z)=m$ after crossover of $x\in S(\top,n)$ with $y\in S(\top,k)\cap S(x,h)$ into $z\in S(\top,m)$ is minimized if the recombination radius $r\in[0,l]$ has the values $r=0$ for $n<k$, $r=l$ for $n>k$, and any value for $n=k$.
  \label{pr:recombination-step}
\end{proposition}

\begin{proof}
  The minimization $\bE_P\{m\mid n,k,h,r\}<n$ over the recombination radius $r\in[0,l]$ follows from equation~(\ref{eq:e-recombination}).
\end{proof}

This simple result implies the following recombination rate function for the uniform crossover operator (Example~\ref{ex:uniform-crossover}):
\[
  \hat\nu(n,k)=
  \begin{cases}
    0&\mbox{ if $n<k$}\\
    1/2&\mbox{ if $n=k$}\\
    1&\mbox{ if $n>k$}
  \end{cases}\,.
\]
The application of such a strategy for recombination can be limited if the population has no individuals with equal distances $d(\top,x)=d(\top,y)$.  Another approach is to maximize the probability of crossover recombination directly into optimum.  As for mutation, this problem has exact solution.

\begin{proposition}[Recombination into the optimum]
  The probability $P(m=0\mid n,k,h,r)$ that crossover recombination of string $x\in S(\top,n)\subset\cH_\alpha^l$ with $y\in S(\top,k)\cap S(x,h)$ results in string $z=\top$ is
  \[
    P(m=0\mid n,k,h,r)=\frac{|\{\top\in I(x,y,r)\}|}{|I(x,y,r)|}=
    \begin{cases}
      \frac{\binom{l-h}{r-n}}{\binom{l}{r}} & \text{if $h=n+k\leq l$, $n\leq r\leq l-k$}\\
      0 & \text{otherwise}
    \end{cases}
    \,,
  \]
  where $|\{\top \in I(x, y, r)\}|$ is the number of copies of element $\top$ in the recombination potential $I(x,y,r)$ (recall that it is a multiset); $h=d(x,y)$ is recombination capacity, and $r\in[0,l]$ is recombination radius.  The optimal recombination radius maximizing the above probability is
  \begin{equation}
    \hat r=\biggl\lfloor l\,\left(\frac{n}{n+k}\right)\biggr\rceil\,,\label{eq:r-optimal}
  \end{equation}
  where $n=d(\top,x)$ and $k=d(\top,y)$ (here $\lfloor\cdot\rceil$ denotes the nearest integer).
  \label{pr:recombination-optimum}
\end{proposition}

\begin{proof}
  Crossover recombination into $z=\top$, $d(\top,z)=m=0$, implies that the recombination radius is $r\geq n=d(\top,x)$ letters, of which there should be exactly $r_+=n$ beneficial substitutions, and which is also their maximum number $h_+=n$.  For deleterious substitutions $r_-=0$ and $h_-=h_+-(n-k)=k=d(\top,y)$ by equation~(\ref{eq:nr-difference}).  There can be no neutral substitutions among $h=d(x,y)$ different letters, which implies $h_0=0$ and $h=h_++h_-=n+k$.  Thus, recombination capacity $h=d(x,y)$ must be maximized (this also maximizes the variance of distances after crossover (\ref{eq:var-recombination})).  There can be $r_0=r-r_+-r_-=r-n\geq0$ neutral substitutions among $l-d(x,y)=l-h$ identical letters, which also must be identical to the corresponding letters in $\top$.  Because $r_0=r-n\leq l-h=l-n-k$, we also obtain the upper bound $r\leq l-k$.  The formula for the probability $P(m=0\mid n,k,h,r)$ is obtained by substituting the values $r_+=h_+=n$, $r_-=0$, $r_0=r-n$, $h_-=k$, $h=h_++h_-$ into formula~(\ref{eq:p-recombination}) and considering the constraints $h=d(x,y)=n+k\leq l$ and $n\leq r\leq l-k$.  This probability is the proportion of all optimal elements $\top$ in the recombination potential $I(x,y,r)$, which is a multiset (i.e. it may contain multiple copies of $\top$).

  The optimal recombination radius $r\in[n,l-k]$ maximizing the probability can be found by setting its derivative over $r$ to zero.  Using the following formula for the derivative of the binomial coefficient \cite{wolfram01:_binomial}:
  \[
    \frac{\partial}{\partial r}\binom{l}{r}=\binom{l}{r}[H_{l-r}-H_r]\,,
  \]
  where $H_r$ is the $r$th harmonic number, and employing simple approximation $H_r\approx\ln r$ gives the following necessary optimality condition:
  \begin{align*}
    \frac{\partial}{\partial r}\frac{\binom{l-n-k}{r-n}}{\binom{l}{r}}
    &=\frac{\binom{l-n-k}{\hat r-n}}{\binom{l}{\hat r}}\left[H_{l-k-\hat r}-H_{\hat r-n}-H_{l-\hat r}+H_{\hat r}\right]\\
    &\approx\frac{\binom{l-n-k}{\hat r-n}}{\binom{l}{\hat r}}\left[\ln\frac{(l-k-\hat r)\hat r}{(\hat r-n)(l-\hat r)}\right]=0\,.
  \end{align*}
  The root $\hat r$ of the above equation is obtained by setting the logarithm to zero resulting in the following equations:
  \[
    (l-k-\hat r)\hat r = (\hat r-n)(l-\hat r)\quad\implies\quad
    -k\hat r = -n(l - \hat r)\quad\implies\quad
    \hat r=l\left(\frac{n}{n+k}\right)\,.
  \]
  One can check also that the second derivative is negative for $\hat r$:
  \[
    \frac{\partial^2}{\partial r^2}\frac{\binom{l-n-k}{r-n}}{\binom{l}{r}}
    \approx\frac{\binom{l-n-k}{\hat r-n}}{\binom{l}{\hat r}}\left[\left(\ln\frac{(l-k-\hat r)\hat r}{(\hat r-n)(l-\hat r)}\right)^2 + \frac{\partial}{\partial r}\ln\frac{(l-k-\hat r)\hat r}{(\hat r-n)(l-\hat r)}\right]\,.
  \]
  The square of the logarithm on the left is zero at $\hat r$.  The derivative of the logarithm on the right is
  \begin{align*}
    \frac{\partial}{\partial r}\ln\frac{(l-k-\hat r)\hat r}{(\hat r-n)(l-\hat r)}
    &= \left.-\frac{1}{l-k-\hat r} + \frac{1}{l-\hat r} - \frac{1}{\hat r - n} + \frac{1}{\hat r}\right|_{\hat r = l\left(\frac{n}{n+k}\right)}\\
    &=(n+k)\left[\frac{1}{n}+\frac{1}{k}\right]\left[\frac{1}{l}-\frac{1}{l-(n+k)}\right]\leq0\,,
  \end{align*}
  because $l\geq l-(n+k)$.  Formula~(\ref{eq:r-optimal}) is the nearest integer of $l\left(\frac{n}{n+k}\right)$.
\end{proof}

The procedure for maximizing the probability $P(m=0\mid n,k,h,r)$ of crossover recombination into the optimum $\top$ can now be outlined:
\begin{enumerate}
\item Let $x\in S(\top,n)$ be the first parent string (i.e. the distance $d(\top,x)=n$ is fixed).
\item Choose the set $\{y\in S(\top,k)\}$ of second parents as close as possible to the optimum $\top$, because decreasing the distance $d(\top,y)=k$ increases the numerator $\binom{l-h}{r-n}$ in the probability $P(m=0\mid n,k,h,r)$.
\item Choose the second parent $y\in S(\top,k)$ with the maximum recombination capacity $d(x,y)=h=n+k\leq l$ (otherwise, the probability $P(m=0\mid n,k,h,r)$ is zero).
\item Recombine precisely $\hat r=\biggl\lfloor l\,\left(\frac{n}{n+k}\right)\biggr\rceil$ letters to maximize $P(m=0\mid n,k,h,r)$.
\end{enumerate}

One can see from formula~(\ref{eq:r-optimal}) that the optimal recombination radius increases with distance $d(\top,x)=n$ of the first string and decreases with distance $d(\top,y)=k$ of the second string.  Observe also that at equal distances $n=k$ the optimal value is $\hat r=l/2$, which also maximizes the distance variance~(\ref{eq:var-recombination}).

In the end of this section, let us compare two recombination operators --- the uniform (Example~\ref{ex:uniform-crossover}) and the one-point crossover (Example~\ref{ex:one-point-crossover}).  In the case of uniform crossover, the recombination radius is a binomial random variable, and taking into account the conditions $h=n+k$ and $r\in[n,l-k]$ we obtain
\[
  P_\nu(m=0\mid n,k) = \sum_{r=n}^{l-k}\binom{l-n-k }{r-n}\nu^r(n,k)[1-\nu(n,k)]^{l-r}\,.
\]
Optimization of the recombination rate is complicated due to the range of possible recombination radii $r\in[n,l-k]$.  For simplicity, let us assume that $r$ takes only one value $r\in[n,l-k]$.  In this case, the maximum is found by differentiation:
\begin{align*}
  \frac{\partial}{\partial\nu}P_\nu&=\binom{l-n-k }{r-n}\hat\nu^r[1-\hat\nu]^{l-r}
  \underbrace{\left(\frac{r}{\hat\nu} - \frac{l-r}{1-\hat\nu}\right)}_{=0}
  =0\qquad\implies\qquad
  \hat\nu=\frac{r}{l}\,,\\
  \frac{\partial^2}{\partial\nu^2}P_\nu&=\binom{l-n-k }{r-n}\hat\nu^r[1-\hat\nu]^{l-r}
   \left[\underbrace{\left(\frac{r}{\hat\nu} - \frac{l-r}{1-\hat\nu}\right)^2}_{=0} - \frac{l - 2r/\hat\nu + r/\hat\nu^2}{(1-\hat\nu)^2}\right] \\
                                   &=\binom{l-n-k }{r-n}\hat\nu^r[1-\hat\nu]^{l-r}
                                     \left[ \frac{l(1-l/r)}{(1-r/l)^2}\right]\leq 0\,,
\end{align*}
because $1-l/r\leq0$.  Substituting the optimal recombination radius~(\ref{eq:r-optimal}) results in the following recombination rate function:
\[
  \hat\nu(n,k)=\frac{n}{n+k}\,.
\]
Thus, uniform crossover with the above recombination rate makes the mean value $\bE\{r\}=l\nu$ of the recombination radius equal to the optimal value~(\ref{eq:r-optimal}).  However, its variance $\sigma^2(r)=l\nu(1-\nu)$ is generally not zero, meaning that the exact optimal value of recombination radius is not guaranteed.

In the case of one-point crossover, the distribution of recombination radius is the Dirac $\delta_{l-i}(r)$ with zero variance, where $i\in[1,\ldots,l]$ is the index of one-point crossover.  Therefore, if the index of one-point crossover is set to $i=l-\hat r=\lfloor lk/(n+k)\rceil$, then it is feasible to guarantee the optimal value given by equation~(\ref{eq:r-optimal}).  This could be the advantage of one-point crossover over the uniform crossover.  Interestingly, a process similar to one-point crossover is used in nature to exchange genetic material between pairs of homologous non-sister chromatids.

\section{Discussion}

We have analysed geometry and combinatorics of mutation and crossover operators in Hamming spaces.  The new formula for geometric probability~(\ref{eq:p-recombination}) of crossover recombination of two strings onto a sphere around an optimum now complements a similar formula~(\ref{eq:p-mutation}) for mutation that was derived previously in~\cite{Belavkin_etal11:_ecal11,Belavkin11:_itw11,Belavkin11:_dyninf,Belavkin11:_qbic11}.  Combined with the information about specific mutation and crossover operators (e.g. Examples~\ref{ex:point-mutation}, \ref{ex:uniform-crossover}, \ref{ex:one-point-crossover}) one can compute stochastic matrices $M$ and $R$ to represent Markov operators, defined by equations~(\ref{eq:operator-mutation}) and (\ref{eq:operator-recombination}), transforming distance distributions $p(t):=P_t\{x\in S(\top,n)\}$ under mutation and recombination in a Hamming space $\cH_\alpha^l$.  Their product together with the diagonal $(l+1)\times(l+1)$ matrix $S$ representing selection gives complete Markov evolution of distance distributions:
\[
  p(t)=(MRS)^t\,p(0)\,,\qquad p(0)=P_0\,.
\]
This opens up the possibility for computer simulations and numerical optimization of long-term evolutionary dynamics under various control functions and strategies, such as the variable mutation rates $\mu(n)$ (e.g. as in equation~(\ref{eq:linear-mutation-rate})), recombination radii $r(n,k)$ (e.g. as in equation~(\ref{eq:r-optimal})) and pairing strategies (e.g. Examples~\ref{ex:matching}, \ref{ex:pairing-h-random}, \ref{ex:pairing-h-specific}).  In some cases, analytic solutions are also possible, such as the optimality conditions given in Propositions~\ref{pr:mutation-step} and \ref{pr:mutation-optimum} for mutation (previously presented in \cite{Belavkin11:_itw11,Belavkin11:_dyninf,Belavkin11:_qbic11}) and in Propositions~\ref{pr:recombination-step} and \ref{pr:recombination-optimum} for crossover recombination.  It is important to note, however, that such solutions that are optimal for these specific criteria may not be optimal for other criteria (e.g. see the discussion in \cite{Buskulic-Doer19} or various optimality criteria and constraints in \cite{Belavkin11:_dyninf}).  Simulations using the above Markov process can be used for optimization of evolutionary dynamics over multiple generations and considering other characteristics, such as the rate of convergence or the running time.  The latter can be estimated as time to absorption for stochastic matrix $MRS$ by considering the optimum $\top\in\cH_\alpha^l$ as an absorbing state.  Such a programme has already been realized in \cite{Belavkin11:_dyninf}, but only for the mutation operator $M$.  This work extends the range of tools suitable for a more complete study.  Future work may consider potential applications to the run-time analysis of evolutionary algorithms and optimization of mutation and recombination operators in more complex fitness landscapes (i.e. when fitness is not the negative distance to optimum).  Such analysis can be facilitated by the formulae derived here and considering monotonic relations between fitness and distance that can often be postulated \cite{Belavkin_etal16:_jomb}.

Another interesting direction to explore is the interaction between mutation and crossover operators.  Our analysis suggests that mutation and crossover recombination may have different and in some sense complementary properties.  Mutation has the advantage that its range is the entire space $\{1,\ldots,\alpha\}^l$ of strings.  However, it lacks direction, and when strings are closer to an optimum the majority of mutations are deleterious.  Maximization of the probability of beneficial mutation requires that mutation rates decrease as strings evolve closer to a fitness peak, which has the inevitable effect of slowing down the evolution.  Recombination, on the other hand, acts in a subspace defined by the current population.  However, unlike mutation, recombination can have a direction towards higher fitness, and it equalizes the chances of beneficial and deleterious recombinations.  It has properties, such as variance of distance distribution, that are translation invariant.  These observations suggest that mutation can be more important for diversity and adaptation of the population that is far away from the fitness peak.  Once the population has evolved closer to the fitness peak and the mutation rate reduced, recombination may become more important to maintain the rate of adaptation.  These hypotheses can be tested using simulations.

It is important to emphasize that the theory and probability formulae presented here are exact and not approximate or asymptotic.  At the same time the model concerns a Markov process with only $l+1$ states (i.e. the range $\{0,\ldots,l\}$ of Hamming distance on $\cH_\alpha^l$) and square $(l+1)\times(l+1)$ matrices.  Even though $l$ can be large in some cases, this model is more computationally tractable than other approaches, such as \cite{Nix-Vose92,Vafaee_etal10} that used Markov processes with states corresponding to all possible variable size populations of strings.  Furthermore, many properties can be derived from simulations with small $l$.  In addition, the formulae for the first two moments (i.e. equations (\ref{eq:e-mutation}, \ref{eq:e-recombination}) for the means and  (\ref{eq:var-mutation}, \ref{eq:var-recombination}) for the variances) can be used to infer some approximate properties.  All formulae are valid for strings with arbitrary alphabet size $\alpha\in\bN$ broadening their scope of applications to different areas including biological systems.

Although the analysis presented here is based on information about Hamming distances between strings, the conclusions can be translated into more practical or biologically relevant notions of fitness and similarity.  Previously we showed that fitness is related to distance from an optimum at least in some neighbourhoods of a local optimum in a broad class of fitness landscapes \cite{Belavkin_etal16:_jomb}.  Theoretical predictions about optimal mutation rates were tested in computational experiments with transcription factor binding landscapes as well as experiments \emph{in vivo} with various microbes \cite{eids_nature14,Krasovec_etal17:_plos,Krasovec18:_isme}.  They discovered that mutation rates are strongly anticorrelated with population density, which microbes can sense and that is related to biological fitness (i.e. the replication rate).  Thus, organisms may use a control strategy of mutation parameters similar to that predicted by the theory in order to increase their adaptability.  This trait appears to exist across all domains of life \cite{Krasovec_etal17:_plos}.  It is reasonable to assume that similar strategies may exist for controlling parameters of recombination operators.  Testing these hypotheses experimentally is an exciting prospect.

\backmatter

\bmhead{Acknowledgements}
The author deeply acknowledges Alastair Channon, John Aston, Christopher Knight, Rok Kra{\v s}ovec, Elizabeth Aston and Danna Gifford for their contributions to previous joint work and collaborations on evolutionary dynamics of mutation and discussions of the new results on crossover recombination.  Anton Eremeev is acknowledged for discussing early drafts of this work in 2016.  Professors Satoru Miyazaki is deeply acknowledged for inviting and hosting the author at the Bioinformatics seminars in Tokyo University of Science between 2012--2024, where results of this work were presented and discussed on multiple occasions.  The author is grateful to Professor Keiko Sato for additional discussions of mathematical results.  The author deeply acknowledges very useful comments of three anonymous reviewers.

\section*{Declarations}


\subsection*{Funding}

This work was supported in part by \grant.







\begin{thebibliography}{38}
\ifx \bisbn   \undefined \def \bisbn  #1{ISBN #1}\fi
\ifx \binits  \undefined \def \binits#1{#1}\fi
\ifx \bauthor  \undefined \def \bauthor#1{#1}\fi
\ifx \batitle  \undefined \def \batitle#1{#1}\fi
\ifx \bjtitle  \undefined \def \bjtitle#1{#1}\fi
\ifx \bvolume  \undefined \def \bvolume#1{\textbf{#1}}\fi
\ifx \byear  \undefined \def \byear#1{#1}\fi
\ifx \bissue  \undefined \def \bissue#1{#1}\fi
\ifx \bfpage  \undefined \def \bfpage#1{#1}\fi
\ifx \blpage  \undefined \def \blpage #1{#1}\fi
\ifx \burl  \undefined \def \burl#1{\textsf{#1}}\fi
\ifx \doiurl  \undefined \def \doiurl#1{\url{https://doi.org/#1}}\fi
\ifx \betal  \undefined \def \betal{\textit{et al.}}\fi
\ifx \binstitute  \undefined \def \binstitute#1{#1}\fi
\ifx \binstitutionaled  \undefined \def \binstitutionaled#1{#1}\fi
\ifx \bctitle  \undefined \def \bctitle#1{#1}\fi
\ifx \beditor  \undefined \def \beditor#1{#1}\fi
\ifx \bpublisher  \undefined \def \bpublisher#1{#1}\fi
\ifx \bbtitle  \undefined \def \bbtitle#1{#1}\fi
\ifx \bedition  \undefined \def \bedition#1{#1}\fi
\ifx \bseriesno  \undefined \def \bseriesno#1{#1}\fi
\ifx \blocation  \undefined \def \blocation#1{#1}\fi
\ifx \bsertitle  \undefined \def \bsertitle#1{#1}\fi
\ifx \bsnm \undefined \def \bsnm#1{#1}\fi
\ifx \bsuffix \undefined \def \bsuffix#1{#1}\fi
\ifx \bparticle \undefined \def \bparticle#1{#1}\fi
\ifx \barticle \undefined \def \barticle#1{#1}\fi
\bibcommenthead
\ifx \bconfdate \undefined \def \bconfdate #1{#1}\fi
\ifx \botherref \undefined \def \botherref #1{#1}\fi
\ifx \url \undefined \def \url#1{\textsf{#1}}\fi
\ifx \bchapter \undefined \def \bchapter#1{#1}\fi
\ifx \bbook \undefined \def \bbook#1{#1}\fi
\ifx \bcomment \undefined \def \bcomment#1{#1}\fi
\ifx \oauthor \undefined \def \oauthor#1{#1}\fi
\ifx \citeauthoryear \undefined \def \citeauthoryear#1{#1}\fi
\ifx \endbibitem  \undefined \def \endbibitem {}\fi
\ifx \bconflocation  \undefined \def \bconflocation#1{#1}\fi
\ifx \arxivurl  \undefined \def \arxivurl#1{\textsf{#1}}\fi
\csname PreBibitemsHook\endcsname

\bibitem[\protect\citeauthoryear{Fisher}{1930}]{Fisher30}
\begin{bbook}
\bauthor{\bsnm{Fisher}, \binits{R.A.}}:
\bbtitle{The Genetical Theory of Natural Selection}.
\bpublisher{Oxford University Press},
\blocation{Oxford}
(\byear{1930})
\end{bbook}
\endbibitem

\bibitem[\protect\citeauthoryear{Belavkin
  et~al.}{2011}]{Belavkin_etal11:_ecal11}
\begin{bchapter}
\bauthor{\bsnm{Belavkin}, \binits{R.V.}},
\bauthor{\bsnm{Channon}, \binits{A.}},
\bauthor{\bsnm{Aston}, \binits{E.}},
\bauthor{\bsnm{Aston}, \binits{J.}},
\bauthor{\bsnm{Knight}, \binits{C.G.}}:
\bctitle{Theory and practice of optimal mutation rate control in {H}amming
  spaces of {DNA} sequences}.
In: \beditor{\bsnm{Lenaerts}, \binits{T.}},
\beditor{\bsnm{Giacobini}, \binits{M.}},
\beditor{\bsnm{Bersini}, \binits{H.}},
\beditor{\bsnm{Bourgine}, \binits{P.}},
\beditor{\bsnm{Dorigo}, \binits{M.}},
\beditor{\bsnm{Doursat}, \binits{R.}} (eds.)
\bbtitle{Advances in Artificial Life, {ECAL} 2011: {P}roceedings of the 11th
  {E}uropean {C}onference on the {S}ynthesis and {S}imulation of {L}iving
  {S}ystems},
pp. \bfpage{85}--\blpage{92}.
\bpublisher{MIT Press},
\blocation{Cambridge, MA, USA}
(\byear{2011})
\end{bchapter}
\endbibitem

\bibitem[\protect\citeauthoryear{Belavkin}{2011}]{Belavkin11:_itw11}
\begin{bchapter}
\bauthor{\bsnm{Belavkin}, \binits{R.V.}}:
\bctitle{Mutation and optimal search of sequences in nested {H}amming spaces}.
In: \bbtitle{2011 {IEEE} {I}nformation {T}heory {W}orkshop},
pp. \bfpage{90}--\blpage{94}
(\byear{2011})
\end{bchapter}
\endbibitem

\bibitem[\protect\citeauthoryear{Belavkin}{2012}]{Belavkin11:_dyninf}
\begin{bchapter}
\bauthor{\bsnm{Belavkin}, \binits{R.V.}}:
\bctitle{Dynamics of information and optimal control of mutation in
  evolutionary systems}.
In: \beditor{\bsnm{Sorokin}, \binits{A.}},
\beditor{\bsnm{Murphey}, \binits{R.}},
\beditor{\bsnm{Thai}, \binits{M.T.}},
\beditor{\bsnm{Pardalos}, \binits{P.M.}} (eds.)
\bbtitle{Dynamics of Information Systems: Mathematical Foundations}.
\bsertitle{Springer Proceedings in Mathematics and Statistics},
vol. \bseriesno{20},
pp. \bfpage{3}--\blpage{21}.
\bpublisher{Springer},
\blocation{Switzerland}
(\byear{2012})
\end{bchapter}
\endbibitem

\bibitem[\protect\citeauthoryear{Belavkin}{2013}]{Belavkin11:_qbic11}
\begin{bchapter}
\bauthor{\bsnm{Belavkin}, \binits{R.V.}}:
\bctitle{Minimum of information distance criterion for optimal control of
  mutation rate in evolutionary systems}.
In: \beditor{\bsnm{Accardi}, \binits{L.}},
\beditor{\bsnm{Freudenberg}, \binits{W.}},
\beditor{\bsnm{Ohya}, \binits{M.}} (eds.)
\bbtitle{{Q}uantum {B}io-{I}nformatics {V}}.
\bsertitle{{QP}-{PQ}: Quantum Probability and White Noise Analysis},
vol. \bseriesno{30},
pp. \bfpage{95}--\blpage{115}.
\bpublisher{World Scientific},
\blocation{Singapore}
(\byear{2013})
\end{bchapter}
\endbibitem

\bibitem[\protect\citeauthoryear{Jones and
  Forrest}{1995}]{Jones-Forrest95:_fdc}
\begin{bchapter}
\bauthor{\bsnm{Jones}, \binits{T.}},
\bauthor{\bsnm{Forrest}, \binits{S.}}:
\bctitle{Fitness distance correlation as a measure of problem difficulty for
  genetic algorithms}.
In: \beditor{\bsnm{Eshelman}, \binits{L.}} (ed.)
\bbtitle{Proceedings of the {S}ixth {I}nternational {C}onference on {G}enetic
  {A}lgorithms},
\bconflocation{San Francisco, CA},
pp. \bfpage{184}--\blpage{192}
(\byear{1995})
\end{bchapter}
\endbibitem

\bibitem[\protect\citeauthoryear{Poli and Galvan-Lopez}{2012}]{Poli-Galvan12}
\begin{barticle}
\bauthor{\bsnm{Poli}, \binits{R.}},
\bauthor{\bsnm{Galvan-Lopez}, \binits{E.}}:
\batitle{The effects of constant and bit-wise neutrality on problem hardness,
  fitness distance correlation and phenotypic mutation rates}.
\bjtitle{{IEEE} Transactions on Evolutionary Computation}
\bvolume{16}(\bissue{2}),
\bfpage{279}--\blpage{300}
(\byear{2012})
\end{barticle}
\endbibitem

\bibitem[\protect\citeauthoryear{Belavkin et~al.}{2016}]{Belavkin_etal16:_jomb}
\begin{barticle}
\bauthor{\bsnm{Belavkin}, \binits{R.V.}},
\bauthor{\bsnm{Channon}, \binits{A.}},
\bauthor{\bsnm{Aston}, \binits{E.}},
\bauthor{\bsnm{Aston}, \binits{J.}},
\bauthor{\bsnm{Kra\v{s}ovec}, \binits{R.}},
\bauthor{\bsnm{Knight}, \binits{C.G.}}:
\batitle{Monotonicity of fitness landscapes and mutation rate control}.
\bjtitle{Journal of Mathematical Biology}
\bvolume{73}(\bissue{6}),
\bfpage{1491}--\blpage{1524}
(\byear{2016})
\end{barticle}
\endbibitem

\bibitem[\protect\citeauthoryear{Badis et~al.}{2009}]{Badis09}
\begin{barticle}
\bauthor{\bsnm{Badis}, \binits{G.}},
\bauthor{\bsnm{Berger}, \binits{M.F.}},
\bauthor{\bsnm{Philippakis}, \binits{A.A.}},
\bauthor{\bsnm{Talukder}, \binits{S.}},
\bauthor{\bsnm{Gehrke}, \binits{A.R.}},
\bauthor{\bsnm{Jaeger}, \binits{S.A.}},
\bauthor{\bsnm{Chan}, \binits{E.T.}},
\bauthor{\bsnm{Metzler}, \binits{G.}},
\bauthor{\bsnm{Vedenko}, \binits{A.}},
\bauthor{\bsnm{Chen}, \binits{X.}},
\bauthor{\bsnm{Kuznetsov}, \binits{H.}},
\bauthor{\bsnm{Wang}, \binits{C.F.}},
\bauthor{\bsnm{Coburn}, \binits{D.}},
\bauthor{\bsnm{Newburger}, \binits{D.E.}},
\bauthor{\bsnm{Morris}, \binits{Q.}},
\bauthor{\bsnm{Hughes}, \binits{T.R.}},
\bauthor{\bsnm{Bulyk}, \binits{M.L.}}:
\batitle{Diversity and complexity in {DNA} recognition by transcription
  factors}.
\bjtitle{Science}
\bvolume{324}(\bissue{5935}),
\bfpage{1720}--\blpage{3}
(\byear{2009})
\end{barticle}
\endbibitem

\bibitem[\protect\citeauthoryear{Kra\v{s}ovec et~al.}{2014}]{eids_nature14}
\begin{botherref}
\oauthor{\bsnm{Kra\v{s}ovec}, \binits{R.}},
\oauthor{\bsnm{Belavkin}, \binits{R.V.}},
\oauthor{\bsnm{Aston}, \binits{J.A.D.}},
\oauthor{\bsnm{Channon}, \binits{A.}},
\oauthor{\bsnm{Aston}, \binits{E.}},
\oauthor{\bsnm{Rash}, \binits{B.M.}},
\oauthor{\bsnm{Kadirvel}, \binits{M.}},
\oauthor{\bsnm{Forbes}, \binits{S.}},
\oauthor{\bsnm{Knight}, \binits{C.G.}}:
Mutation rate plasticity in rifampicin resistance depends on escherichia coli
  cell-cell interactions.
Nature Communications
\textbf{5}(3742)
(2014)
\end{botherref}
\endbibitem

\bibitem[\protect\citeauthoryear{Kra\v{s}ovec
  et~al.}{2017}]{Krasovec_etal17:_plos}
\begin{botherref}
\oauthor{\bsnm{Kra\v{s}ovec}, \binits{R.}},
\oauthor{\bsnm{Richards}, \binits{H.}},
\oauthor{\bsnm{Gifford}, \binits{D.R.}},
\oauthor{\bsnm{Hatcher}, \binits{C.}},
\oauthor{\bsnm{Faulkner}, \binits{K.J.}},
\oauthor{\bsnm{Belavkin}, \binits{R.V.}},
\oauthor{\bsnm{Channon}, \binits{A.}},
\oauthor{\bsnm{Aston}, \binits{E.}},
\oauthor{\bsnm{McBain}, \binits{A.J.}},
\oauthor{\bsnm{Knight}, \binits{C.G.}}:
Spontaneous mutation rate is a plastic trait associated with population density
  across domains of life.
PLoS Biology
\textbf{15}(8)
(2017)
\end{botherref}
\endbibitem

\bibitem[\protect\citeauthoryear{Kra\v{s}ovec et~al.}{2018}]{Krasovec18:_isme}
\begin{barticle}
\bauthor{\bsnm{Kra\v{s}ovec}, \binits{R.}},
\bauthor{\bsnm{Richards}, \binits{H.}},
\bauthor{\bsnm{Gifford}, \binits{D.R.}},
\bauthor{\bsnm{Belavkin}, \binits{R.V.}},
\bauthor{\bsnm{Channon}, \binits{A.}},
\bauthor{\bsnm{Aston}, \binits{E.}},
\bauthor{\bsnm{McBain}, \binits{A.J.}},
\bauthor{\bsnm{Knight}, \binits{C.G.}}:
\batitle{Opposing effects of final population density and stress on escherichia
  coli mutation rate}.
\bjtitle{The ISME Journal}
\bvolume{12},
\bfpage{2981}--\blpage{2987}
(\byear{2018})
\end{barticle}
\endbibitem

\bibitem[\protect\citeauthoryear{Yang}{2010}]{Yang10}
\begin{bbook}
\bauthor{\bsnm{Yang}, \binits{X.-S.}}:
\bbtitle{Nature-inspired Metaheuristic Algorithms}.
\bpublisher{Luniver Press},
\blocation{Frome, UK}
(\byear{2010})
\end{bbook}
\endbibitem

\bibitem[\protect\citeauthoryear{B{\"a}ck}{1993}]{Back93}
\begin{bchapter}
\bauthor{\bsnm{B{\"a}ck}, \binits{T.}}:
\bctitle{Optimal mutation rates in genetic search}.
In: \beditor{\bsnm{Forrest}, \binits{S.}} (ed.)
\bbtitle{Proceedings of the 5th {I}nternational {C}onference on {G}enetic
  {A}lgorithms},
pp. \bfpage{2}--\blpage{8}.
\bpublisher{Morgan Kaufmann},
\blocation{San Francisco, CA}
(\byear{1993})
\end{bchapter}
\endbibitem

\bibitem[\protect\citeauthoryear{Ochoa}{2002}]{Ochoa02}
\begin{bchapter}
\bauthor{\bsnm{Ochoa}, \binits{G.}}:
\bctitle{Setting the mutation rate: Scope and limitations of the $1/l$
  heuristics}.
In: \bbtitle{Proceedings of {G}enetic and {E}volutionary {C}omputation
  {C}onference ({GECCO}-2002)},
pp. \bfpage{315}--\blpage{322}.
\bpublisher{Morgan Kaufmann},
\blocation{San Francisco, CA}
(\byear{2002})
\end{bchapter}
\endbibitem

\bibitem[\protect\citeauthoryear{Fogarty}{1989}]{Fogarty89}
\begin{bchapter}
\bauthor{\bsnm{Fogarty}, \binits{T.C.}}:
\bctitle{Varying the probability of mutation in the genetic algorithm}.
In: \beditor{\bsnm{Schaffer}, \binits{J.D.}} (ed.)
\bbtitle{Proceedings of the 3rd {I}nternational {C}onference on {G}enetic
  {A}lgorithms},
pp. \bfpage{104}--\blpage{109}.
\bpublisher{Morgan Kaufmann},
\blocation{San Francisco, CA}
(\byear{1989})
\end{bchapter}
\endbibitem

\bibitem[\protect\citeauthoryear{Yanagiya}{1993}]{Yanagiya93}
\begin{bchapter}
\bauthor{\bsnm{Yanagiya}, \binits{M.}}:
\bctitle{A simple mutation-dependent genetic algorithm}.
In: \beditor{\bsnm{Forrest}, \binits{S.}} (ed.)
\bbtitle{Proceedings of the 5th {I}nternational {C}onference on {G}enetic
  {A}lgorithms},
p. \bfpage{659}.
\bpublisher{Morgan Kaufmann},
\blocation{San Francisco, CA}
(\byear{1993})
\end{bchapter}
\endbibitem

\bibitem[\protect\citeauthoryear{Srinivas and
  Patnaik}{1994}]{Srinivas-Patnaik94}
\begin{barticle}
\bauthor{\bsnm{Srinivas}, \binits{M.}},
\bauthor{\bsnm{Patnaik}, \binits{L.M.}}:
\batitle{Adaptive probabilities of crossover and mutation in genetic
  algorithms}.
\bjtitle{{IEEE} Transactions on Systems, Man and Cybernetics}
\bvolume{24}(\bissue{4}),
\bfpage{656}--\blpage{667}
(\byear{1994})
\end{barticle}
\endbibitem

\bibitem[\protect\citeauthoryear{Braga and
  Aleksander}{1994}]{Braga-Aleksander94}
\begin{bchapter}
\bauthor{\bsnm{Braga}, \binits{A.D.P.}},
\bauthor{\bsnm{Aleksander}, \binits{I.}}:
\bctitle{Determining overlap of classes in the $n$-dimensional {B}oolean
  space}.
In: \bbtitle{Neural Networks, 1994. {IEEE} {W}orld {C}ongress on
  {C}omputational {I}ntelligence., 1994 {IEEE} {I}nternational {C}onference
  on},
vol. \bseriesno{7},
pp. \bfpage{8}--\blpage{13}
(\byear{1994})
\end{bchapter}
\endbibitem

\bibitem[\protect\citeauthoryear{Eiben et~al.}{1999}]{Eiben_etal99}
\begin{barticle}
\bauthor{\bsnm{Eiben}, \binits{A.E.}},
\bauthor{\bsnm{Hinterding}, \binits{R.}},
\bauthor{\bsnm{Michalewicz}, \binits{Z.}}:
\batitle{Parameter control in evolutionary algorithms}.
\bjtitle{IEEE Transactions on Evolutionary Computation}
\bvolume{3}(\bissue{2}),
\bfpage{124}--\blpage{141}
(\byear{1999})
\end{barticle}
\endbibitem

\bibitem[\protect\citeauthoryear{Falco et~al.}{2002}]{Falco_etal02}
\begin{barticle}
\bauthor{\bsnm{Falco}, \binits{I.D.}},
\bauthor{\bsnm{Cioppa}, \binits{A.D.}},
\bauthor{\bsnm{Tarantino}, \binits{E.}}:
\batitle{Mutation-based genetic algorithm: performance evaluation}.
\bjtitle{Applied Soft Computing}
\bvolume{1}(\bissue{4}),
\bfpage{285}--\blpage{299}
(\byear{2002})
\end{barticle}
\endbibitem

\bibitem[\protect\citeauthoryear{Doerr et~al.}{2011}]{Doerr_etal11}
\begin{bchapter}
\bauthor{\bsnm{Doerr}, \binits{B.}},
\bauthor{\bsnm{Johannsen}, \binits{D.}},
\bauthor{\bsnm{Schmidt}, \binits{M.}}:
\bctitle{Runtime analysis of the $(1+1)$ evolutionary algorithm on strings over
  finite alphabets}.
In: \bbtitle{{FOGA}'11: Proceedings of the 11th Workshop Proceedings on
  {F}oundations of {G}enetic {A}lgorithms},
pp. \bfpage{119}--\blpage{126}
(\byear{2011})
\end{bchapter}
\endbibitem

\bibitem[\protect\citeauthoryear{Doerr and Pohl}{2012}]{Doerr-Pohl12}
\begin{bchapter}
\bauthor{\bsnm{Doerr}, \binits{B.}},
\bauthor{\bsnm{Pohl}, \binits{S.}}:
\bctitle{Run-time analysis of the $(1+1)$ evolutionary algorithm optimizing
  linear functions over a finite alphabet}.
In: \bbtitle{{GECCO}'12: Proceedings of the 14th Annual Conference on {G}enetic
  and {E}volutionary {C}omputation},
pp. \bfpage{1317}--\blpage{1324}
(\byear{2012})
\end{bchapter}
\endbibitem

\bibitem[\protect\citeauthoryear{Doerr et~al.}{2018}]{Doerr_etal18}
\begin{barticle}
\bauthor{\bsnm{Doerr}, \binits{B.}},
\bauthor{\bsnm{Doerr}, \binits{C.}},
\bauthor{\bsnm{K\"{o}tzing}, \binits{T.}}:
\batitle{Static and self-adjusting mutation strengths for multi-valued decision
  variables}.
\bjtitle{Algorithmica}
\bvolume{80},
\bfpage{1732}--\blpage{1768}
(\byear{2018})
\end{barticle}
\endbibitem

\bibitem[\protect\citeauthoryear{K{\"o}tzing
  et~al.}{2011}]{Koetzing11:_crossover}
\begin{bchapter}
\bauthor{\bsnm{K{\"o}tzing}, \binits{T.}},
\bauthor{\bsnm{Sudholt}, \binits{D.}},
\bauthor{\bsnm{Theile}, \binits{M.}}:
\bctitle{How crossover helps in pseudo-boolean optimization}.
In: \bbtitle{GECCO '11: Proceedings of the 13th Annual Conference on Genetic
  and Evolutionary computationJuly 2011},
pp. \bfpage{989}--\blpage{996}
(\byear{2011})
\end{bchapter}
\endbibitem

\bibitem[\protect\citeauthoryear{Dang et~al.}{2018}]{Dang18:_crossover}
\begin{barticle}
\bauthor{\bsnm{Dang}, \binits{D.-C.}},
\bauthor{\bsnm{Friedrich}, \binits{T.}},
\bauthor{\bsnm{Kötzing}, \binits{T.}},
\bauthor{\bsnm{Krejca}, \binits{M.S.}},
\bauthor{\bsnm{Lehre}, \binits{P.K.}},
\bauthor{\bsnm{Oliveto}, \binits{P.S.}},
\bauthor{\bsnm{Sudholt}, \binits{D.}},
\bauthor{\bsnm{Sutton}, \binits{A.M.}}:
\batitle{Escaping local optima using crossover with emergent diversity}.
\bjtitle{IEEE Transactions on Evolutionary Computation}
\bvolume{22}(\bissue{3}),
\bfpage{484}--\blpage{497}
(\byear{2018})
\end{barticle}
\endbibitem

\bibitem[\protect\citeauthoryear{Moraglio et~al.}{2007}]{Moraglio07}
\begin{barticle}
\bauthor{\bsnm{Moraglio}, \binits{A.}},
\bauthor{\bsnm{Kim}, \binits{Y.-H.}},
\bauthor{\bsnm{Yoon}, \binits{Y.}},
\bauthor{\bsnm{Moon}, \binits{B.-R.}}:
\batitle{Geometric crossovers for multiway graph partitioning}.
\bjtitle{Evolutionary Computation}
\bvolume{15}(\bissue{4}),
\bfpage{445}--\blpage{474}
(\byear{2007})
\end{barticle}
\endbibitem

\bibitem[\protect\citeauthoryear{Eremeev}{2000}]{Eremeev99}
\begin{bchapter}
\bauthor{\bsnm{Eremeev}, \binits{A.V.}}:
\bctitle{Modeling and analysis of genetic algorithm with tournament selection}.
In: \beditor{\bsnm{Fonlupt}, \binits{C.}},
\beditor{\bsnm{Hao}, \binits{J.-K.}},
\beditor{\bsnm{Lutton}, \binits{E.}},
\beditor{\bsnm{Schoenauer}, \binits{M.}},
\beditor{\bsnm{Ronald}, \binits{E.}} (eds.)
\bbtitle{Artificial Evolution: 4th European Conference, AE'99, Dunkerque,
  France, November 3-5, 1999. Selected Papers}.
\bsertitle{Lecture Notes in Computer Science},
vol. \bseriesno{1829},
pp. \bfpage{84}--\blpage{95}.
\bpublisher{Springer},
\blocation{Berlin, Heidelberg}
(\byear{2000})
\end{bchapter}
\endbibitem

\bibitem[\protect\citeauthoryear{Eremeev}{2008}]{Eremeev08}
\begin{barticle}
\bauthor{\bsnm{Eremeev}, \binits{A.V.}}:
\batitle{On complexity of optimal recombination for binary representations of
  solutions}.
\bjtitle{Evolutionary Computation}
\bvolume{16}(\bissue{1}),
\bfpage{127}--\blpage{147}
(\byear{2008})
\end{barticle}
\endbibitem

\bibitem[\protect\citeauthoryear{Eremeev}{2011}]{Eremeev2011}
\begin{bchapter}
\bauthor{\bsnm{Eremeev}, \binits{A.V.}}:
\bctitle{On Complexity of the Optimal Recombination for the Travelling Salesman
  Problem}.
In: \beditor{\bsnm{Merz}, \binits{P.}},
\beditor{\bsnm{Hao}, \binits{J.-K.}} (eds.)
\bbtitle{Evolutionary Computation in Combinatorial Optimization: 11th European
  Conference, EvoCOP 2011, Torino, Italy, April 27-29, 2011. Proceedings},
pp. \bfpage{215}--\blpage{225}.
\bpublisher{Springer},
\blocation{Berlin, Heidelberg}
(\byear{2011})
\end{bchapter}
\endbibitem

\bibitem[\protect\citeauthoryear{Eremeev and
  Kovalenko}{2014a}]{Eremeev-Kovalenko14:_1}
\begin{barticle}
\bauthor{\bsnm{Eremeev}, \binits{A.V.}},
\bauthor{\bsnm{Kovalenko}, \binits{J.V.}}:
\batitle{Optimal recombination in genetic algorithms for combinatorial
  optimization problems: Part {I}}.
\bjtitle{Yugoslav Journal of Operations Research}
\bvolume{24}(\bissue{1}),
\bfpage{1}--\blpage{20}
(\byear{2014})
\end{barticle}
\endbibitem

\bibitem[\protect\citeauthoryear{Eremeev and
  Kovalenko}{2014b}]{Eremeev-Kovalenko14:_2}
\begin{barticle}
\bauthor{\bsnm{Eremeev}, \binits{A.V.}},
\bauthor{\bsnm{Kovalenko}, \binits{J.V.}}:
\batitle{Optimal recombination in genetic algorithms for combinatorial
  optimization problems: Part {II}}.
\bjtitle{Yugoslav Journal of Operations Research}
\bvolume{24}(\bissue{2}),
\bfpage{165}--\blpage{186}
(\byear{2014})
\end{barticle}
\endbibitem

\bibitem[\protect\citeauthoryear{Ahlswede and Katona}{1977}]{Ahlswede-Katona77}
\begin{barticle}
\bauthor{\bsnm{Ahlswede}, \binits{R.}},
\bauthor{\bsnm{Katona}, \binits{G.O.H.}}:
\batitle{Contributions to the geometry of {H}amming spaces}.
\bjtitle{Discrete Mathematics}
\bvolume{17}(\bissue{1}),
\bfpage{1}--\blpage{22}
(\byear{1977})
\end{barticle}
\endbibitem

\bibitem[\protect\citeauthoryear{von Neumann and
  Morgenstern}{1944}]{Neumann-Morgenstern}
\begin{bbook}
\bauthor{\bsnm{Neumann}, \binits{J.}},
\bauthor{\bsnm{Morgenstern}, \binits{O.}}:
\bbtitle{Theory of Games and Economic Behavior}.
\bpublisher{Princeton University Press},
\blocation{Princeton, NJ}
(\byear{1944})
\end{bbook}
\endbibitem

\bibitem[\protect\citeauthoryear{}{2001}]{wolfram01:_binomial}
\begin{botherref}
Binomial.
Wolfram Research, Inc.
\url{http://functions.wolfram.com/06.03.20.0003.01}
(2001)
\end{botherref}
\endbibitem

\bibitem[\protect\citeauthoryear{Buskulic and Doerr}{2019}]{Buskulic-Doer19}
\begin{bchapter}
\bauthor{\bsnm{Buskulic}, \binits{N.}},
\bauthor{\bsnm{Doerr}, \binits{C.}}:
\bctitle{Maximizing drift is not optimal for solving onemax}.
In: \beditor{\bsnm{nez}, \binits{M.L.-I.}} (ed.)
\bbtitle{{GECCO}'19: Proceedings of the {G}enetic and {E}volutionary
  {C}omputation {C}onference {C}ompanion},
pp. \bfpage{425}--\blpage{426}
(\byear{2019})
\end{bchapter}
\endbibitem

\bibitem[\protect\citeauthoryear{Nix and Vose}{1992}]{Nix-Vose92}
\begin{barticle}
\bauthor{\bsnm{Nix}, \binits{A.E.}},
\bauthor{\bsnm{Vose}, \binits{M.D.}}:
\batitle{Modeling genetic algorithms with {M}arkov chains}.
\bjtitle{Annals of Mathematics and Artificial Intelligence}
\bvolume{5}(\bissue{1}),
\bfpage{77}--\blpage{88}
(\byear{1992})
\end{barticle}
\endbibitem

\bibitem[\protect\citeauthoryear{Vafaee et~al.}{2010}]{Vafaee_etal10}
\begin{bchapter}
\bauthor{\bsnm{Vafaee}, \binits{F.}},
\bauthor{\bsnm{Tur{\'a}n}, \binits{G.}},
\bauthor{\bsnm{Nelson}, \binits{P.C.}}:
\bctitle{Optimizing genetic operator rates using a {M}arkov chain model of
  genetic algorithms}.
In: \beditor{\bsnm{Pelikan}, \binits{M.}},
\beditor{\bsnm{Branke}, \binits{J.}} (eds.)
\bbtitle{Proceedings of {G}enetic and {E}volutionary {C}omputation {C}onference
  ({GECCO}-2010)},
pp. \bfpage{721}--\blpage{728}.
\bpublisher{Association for Computing Machinery},
\blocation{New York, NY, USA}
(\byear{2010})
\end{bchapter}
\endbibitem

\end{thebibliography}


\begin{appendices}

\section{Proofs}

\subsection{Proof of Lemma~\ref{lm:moments} for the mean and variance of Hamming distance}\label{sec:proof-e-var}

\begin{proof}
  Using the definition of Hamming distance~(\ref{eq:hamming}) as a sum of elementary distances $1-\delta_{x_iy_i}\in\{0,1\}$ we have
  \begin{align*}
    \bE_P\{d(x,y)\}&=\bE_P\left\{\sum_{i=1}^l(1-\delta_{x_iy_i})\right\}=\sum_{i=1}^l\bE_P\{1-\delta_{x_iy_i}\}\,,\\
    \sigma_P^2\{d(x,y)\}&=\bE_P\left\{\left(\sum_{i=1}^l(1-\delta_{x_iy_i})\right)^2\right\}-\left(\bE_P\left\{\sum_{i=1}^l(1-\delta_{x_iy_i})\right\}\right)^2\\
                   &=\sum_{i=1}^l\bE_P\{(1-\delta_{x_iy_i})^2\}+\sum_{i=1}^l\sum_{\substack{j=1\\j\neq i}}^l\bE_P\{(1-\delta_{x_iy_i})(1-\delta_{x_jy_j})\} -\\
                   &\hspace{.5\textwidth} - \left(\sum_{i=1}^l\bE_P\{1-\delta_{x_iy_i}\}\right)^2\,,
  \end{align*}
  where we used additivity of the expected value and the square of the sum formula: $\left(\sum_ia_i\right)^2=\sum_ia_i^2 + 2\sum_{i<j}a_ia_j=\sum_ia_i^2 + \sum_i\sum_{j\neq i}a_ia_j$.   Noticing that $(1-\delta_{x_iy_i})^2=1-\delta_{x_iy_i}$ and denoting the average values of the expectations under summations by $\langle P_i\rangle$ and $\langle P_{ij}\rangle$ (see their definitions in Lemma~\ref{lm:moments}) we obtain the following formulae:
  \begin{align*}
    \bE_P\{d(x,y)\}&=l\langle P_i\rangle\,,\\
    \sigma_P^2\{d(x,y)\}&=l\langle P_i\rangle+l(l-1)\langle P_{ij}\rangle-(l\langle P_i\rangle)^2\,.
  \end{align*}
\end{proof}

The symbol $\langle P_i\rangle$ introduced above is the average (over all positions $i\in\{1,\ldots,l\}$) of the probability that $x_i\neq y_i$.  Indeed, the differences $1-\delta_{x_iy_i}$ are non-zero (and equal to one) only when $x_i\neq y_i$, and therefore the expectations $\bE_P\{1-\delta_{x_iy_i}\}$ are the probabilities that letters $x_i\neq y_i$ at positions $i\in\{1,\ldots,l\}$ (and these probabilities can be different for different $i$).

Similarly, $\langle P_{ij}\rangle$ is the average of joint probability that $x_i\neq y_i$ and $x_j\neq y_j$ at two different positions $i\neq j$ (there are $l(l-1)$ off-diagonal elements).  Indeed, the products $(1-\delta_{x_iy_i})(1-\delta_{x_jy_j})$ are non-zero only when both $x_i\neq y_i$ and $x_j\neq y_j$, and the expectations $\bE_P\{(1-\delta_{x_iy_i})(1-\delta_{x_jy_j})\}$ are the corresponding joint probabilities.

\subsection{Proof of Lemma~\ref{lm:spheres} for intersection of spheres in $\cH_\alpha^l$}\label{sec:proof-spheres}

\begin{proof}
  A substitution of letter $x_i$ at position $i\in\{1,\ldots,l\}$ may result in one of three possibilities:
\begin{itemize}
\item Letter $x_i\neq\top_i$ is substituted to $y_i=\top_i$ resulting in a beneficial substitution.  Such substitutions can only occur in $n=d(\top,x)$ letters $x_i\neq\top_i$.  Denoting by $r_+$ the total number of beneficial substitutions gives $\binom{n}{r_+}$ combinations.
\item Letter $x_i=\top_i$ is substituted to any of $\alpha-1$ letters $y_i\neq\top_i$ resulting in a deleterious substitution.  Such substitutions can only occur in $l-d(\top,x)=l-n$ letters $x_i=\top_i$.  Denoting by $r_-$ the total number of deleterious substitutions gives $(\alpha-1)^{r_-}\binom{l-n}{r_-}$ possibilities.
\item Letter $x_i\neq\top_i$ is substituted to one of the remaining $\alpha-2$ letters $y_i\neq\top_i$ resulting in a neutral substitution.  Such substitutions can only occur in $n=d(\top,x)$ letters $x_i\neq\top_i$ minus the number $r_+$ of beneficial substitutions.  Denoting by $r_0$ the total number of neutral substitutions gives $(\alpha-2)^{r_0}\binom{n-r_+}{r_0}$ possibilities.
\end{itemize}
For specific values of $r_+$, $r_-$ and $r_0$, the total number of combinations is
\[
  (\alpha-2)^{r_0}\binom{n-r_+}{r_0}(\alpha-1)^{r_-}\binom{l-n}{r_-}\binom{n}{r_+}\,.
\]
By equations~(\ref{eq:rm-sum}) and (\ref{eq:rm-difference}), the values of $r_-$ and $r_0$ are related to $r_+\in[0,r]$:
\begin{align*}
&r_-=r_+ -(n-m)\,,\\
&r_0=r-(r_++r_-)=r-2r_++(n-m)\,.
\end{align*}
Formula~(\ref{eq:intersection-sphere}) is obtained by summing over all feasible values of $r_+$ subject to constraints $r_-\geq0$ and $r_0\geq0$.
\end{proof}

\subsection{Proof of Proposition~\ref{pr:mutation}}\label{sec:proof-e-var-mutation}

\begin{proof}
  Here we use formulae~(\ref{eq:e-lp}) and (\ref{eq:var-lp}) with the following average probabilities:
  \begin{align*}
    \langle P_i\rangle&:=\mathbb{P}\{y_i\neq\top_i\mid x\in S(\top,n),y\in S(x,r)\}\,,\\
    \langle P_{ij}\rangle&:=\mathbb{P}\{y_i\neq\top_i\wedge y_j\neq\top_j\mid i\neq j,x\in S(\top,n),y\in S(x,r)\}\,,
  \end{align*}
  where the conditions are defined by the fact that string $y\in S(x,r)\subset\cH_\alpha^l$ is obtained from $x\in S(\top,n)$ by a substitution of $d(x,y)=r$ letters.  Indices $i$ and $j\in\{1,\ldots,l\}$ are positions of letters in the strings $x=(x_1,\ldots,x_l)$, $y=(y_1,\ldots,y_l)$ and $\top=(\top_1,\ldots,\top_l)$.  For each position $i\in\{1,\ldots,l\}$ there are three mutually exclusive possibilities for event $y_i\neq\top_i$:
  \begin{itemize}
  \item Letter $x_i\neq\top_i$ is not substituted, so that $y_i=x_i\neq\top_i$.  The number of letters $x_i\neq\top_i$ is $n=d(\top,x)$, and there are $l-r$ letters that remain not substituted in $y$, which means that the corresponding probability is $(n/l)(1-r/l)$.
  \item Letter $x_i\neq\top_i$ is substituted by $y_i\neq\top_i$, $y_i\neq x_i$.  There are $\alpha-1$ letters in the alphabet not equal to $\top_i$, and there are $\alpha-2$ remaining possibilities for $y_i\neq\top_i$.  Given that there are $n=d(\top,x)$ letters $x_i\neq\top_i$ and $r$ letters are substituted, the corresponding probability is $(n/l)[(\alpha-2)/(\alpha-1)](r/l)$.
  \item Letter $x_i=\top_i$ is substituted by any other letter $y_i\neq\top_i$.  The number of letters $x_i=\top_i$ is $l-n=l-d(\top,x)$, and there are $r$ letters that are substituted, which corresponds to the probability $(1-n/l)(r/l)$.
  \end{itemize}
  Adding the above probabilities for disjoint events gives the desired average probability:
  \[
    \langle P_i\rangle=\frac{n}{l}\left(1-\frac{r}{l}\right)+\frac{n}{l}\left(\frac{\alpha-2}{\alpha-1}\right)\frac{r}{l}+\left(1-\frac{n}{l}\right)\frac{r}{l}\,.
  \]
  Because $\langle P_i\rangle$ is the average probability across all positions $i\in\{1,\ldots,l\}$, its value is the same for all $i$, and the expected value $\bE_P\{m\mid n,r\}$ can be computed as $l\langle P_i\rangle$.  Its expression can be simplified as follows:
  \[
    l\langle P_i\rangle=n\left(1-\frac{1}{\alpha-1}\frac{r}{l}\right)+(l-n)\frac{r}{l}\,.
  \]
  The expression above can also be transformed into formula~(\ref{eq:e-mutation}).

  The average joint probability $\langle P_{ij}\rangle:=\mathbb{P}\{y_i\neq\top_i\wedge y_j\neq\top_j\mid i\neq j,n,r\}$ is factorized into the product $\mathbb{P}\{y_i\neq\top_i\mid n,r\}\mathbb{P}\{y_j\neq\top_j\mid y_i\neq\top_i,i\neq j,n,r\}$, where the first probability for the first event $z_i\neq\top_i$ is the average probability $\langle P_i\rangle$ derived above.  The second is the average conditional probability of the second event $z_j\neq\top_j$ for $j\neq i$ (and conditioned on the first event $z_i\neq\top_i$).  This conditional probability can be determined from the following considerations.

  For each of the three possibilities for $y_i\neq\top_i$ in the string of length $l$, there are three possibilities for $y_j\neq\top_j$ in the remaining string of length $l-1$.  Thus, there are $3\times 3=9$ joint events, the probabilities of which are defined similarly, but using numbers $n-1$ or $n$ and $r-1$ or $r$ depending on the type and position of the first event $y_i\neq\top_i$.  Thus, the product $l(l-1)\langle P_{ij}\rangle$ is as follows:
  \begin{align*}
    l(l-1)\langle P_{ij}\rangle=n\left(1-\frac{r}{l}\right)&\left[(n-1)\left(1-\frac{1}{\alpha-1}\frac{r}{l-1}\right)+(l-n)\frac{r}{l-1}\right]\\
                +n\left(\frac{\alpha-2}{\alpha-1}\right)\frac{r}{l}&\left[(n-1)\left(1-\frac{1}{\alpha-1}\frac{r-1}{l-1}\right)+(l-n)\frac{r-1}{l-1}\right]\\
                +(l-n)\frac{r}{l}&\left[n\left(1-\frac{1}{\alpha-1}\frac{r-1}{l-1}\right)+(l-1-n)\frac{r-1}{l-1}\right]\,.
  \end{align*}
  Formula~(\ref{eq:var-mutation}) for the variance is obtained by substituting the expressions for $l\langle P_i\rangle$ and $l(l-1)\langle P_{ij}\rangle$ into equation $l\langle P_i\rangle+l(l-1)\langle P_{ij}\rangle-(l\langle P_i\rangle)^2$.
\end{proof}

\subsection{Proof of Proposition~\ref{pr:dual}}\label{sec:pr:dual}

\begin{proof}
  Let us denote by $r'=l-r$ the recombination radius of the dual recombination $z'$, and let $r_+'$, $r_-'$ and $r_0'$ be respectively the numbers of beneficial, deleterious and neutral substitutions into $x$.  Because the dual recombination is a substitution of the remaining $l-r$ letters, we have
  \[
    r_++r_+'=h_+\,, \qquad r_-+r_-'=h_-\,.
  \]
  Using the above equations and equations~(\ref{eq:rr-difference}), (\ref{eq:nr-difference}) we have
  \begin{align*}
    r_+'-r_-'&=h_+-h_--(r_+-r_-)\\
            &=(n-k)-(n-m)\,.
  \end{align*}
  On the other hand, by analogy with equation~(\ref{eq:rr-difference}), we have
  \begin{equation}
    r_+'-r_-'=n-m'\,. \label{eq:rr-difference-2}
  \end{equation}
  Therefore, $n-k-(n-m)=m-k=n-m'$.
\end{proof}

\subsection{Proof of Proposition~\ref{pr:recombination}}\label{sec:proof-e-var-recombination}

\begin{proof}
  Here we use formulae~(\ref{eq:e-lp}) and (\ref{eq:var-lp}) with the following average probabilities
  \begin{align*}
    \langle P_i\rangle&:=\mathbb{P}\{z_i\neq\top_i\mid x\in S(\top,n),y\in S(\top,k),d(x,y)=h,r\}\,,\\
    \langle P_{ij}\rangle&:=\mathbb{P}\{z_i\neq\top_i\wedge z_j\neq\top_j\mid i\neq j,x\in S(\top,n),y\in S(\top,k),d(x,y)=h,r\}\,,
  \end{align*}
  where the conditions are defined by the fact that string $z\in S(\top,m)\subset\cH_\alpha^l$ is obtained from $x\in S(\top,n)$ by a substitution of $r\in[0,l]$ letters from string $y\in S(\top,k)\cap S(x,h)$.  Indices $i$ and $j\in\{1,\ldots,l\}$ are positions of letters in the strings.  For each position $i\in\{1,\ldots,l\}$ there are two mutually exclusive possibilities for $z_i\neq\top_i$:
  \begin{itemize}
  \item Letter $x_i\neq\top_i$ is not substituted by $y_i$, so that in the offspring $z_i=x_i\neq\top_i$.  The number of letters $x_i\neq\top_i$ is $n=d(\top,x)$, and there are $l-r$ letters that remain not substituted in $z$, which means that the corresponding probability is $(n/l)(1-r/l)$.
  \item Letter $x_i=\top_i$ is substituted by letter $y_i\neq\top_i$, so that in the offspring $z_i=y_i\neq\top_i$.  The number of letters $y_i\neq\top_i$ is $k=(\top,y)$, and there are $r$ letters that are substituted, which corresponds to probability $(k/l)(r/l)$.
  \end{itemize}
  Adding the above probabilities for disjoint events gives the desired average probability:
  \[
    \langle P_i\rangle=\frac{n}{l}\left(1-\frac{r}{l}\right)+\frac{k}{l}\frac{r}{l}\,.
  \]
  Because $\langle P_i\rangle$ is the average probability across all positions $i\in\{1,\ldots,l\}$, its value is the same for all $i$, and the expected value $\bE_P\{m\mid n,k,h,r\}$ can be computed as $l\langle P_i\rangle$:
  \[
    l\langle P_i\rangle=n\left(1-\frac{r}{l}\right)+k\frac{r}{l}\,.
  \]
  The above expression gives formula~(\ref{eq:e-recombination}).

  The average joint probability $\langle P_{ij}\rangle:=\mathbb{P}\{z_i\neq\top_i\wedge z_j\neq\top_j\mid i\neq j,n,k,h,r\}$ is factorized into the product $\mathbb{P}\{z_i\neq\top_i\mid n,k,h,r\}\mathbb{P}\{z_j\neq\top_j\mid z_i\neq\top_i,i\neq j,n,k,h,r\}$, where the first probability for the first event $z_i\neq\top_i$ is the average probability $\langle P_i\rangle$ derived above.  The second is the average conditional probability of the second event $z_j\neq\top_j$ for $j\neq i$ (and conditioned on the first event $z_i\neq\top_i$).  This conditional probability can be determined from the following considerations.

  Let us decompose each of the two cases of the first event $z_i\neq\top_i$ into three subcases resulting in $2\times 3 = 6$ mutually exclusive subcases of event $z_i\neq\top_i$.  These subcases were not considered for the probability $\langle P_i\rangle$, because it concerns only one index $i$.  When two indices $i$ and $j\neq i$ are considered for the joint probability $\langle P_{ij}\rangle$, these subcases are important, because they influence the numbers that are required for the probability of the second event $z_j\neq\top_j$.

  First, let us consider when any of $n$ letters $x_i\neq\top_i$ are not substituted by $y_i$.  There are three subcases for such non-substitutions.  They can be among
  \begin{enumerate}
  \item $h_+$ letters $y_i=\top_i$ (i.e. some of $h_+$ possible beneficial substitutions do not occur).
  \item $h_0$ letters $y_i\neq\top_i$, $y_i\neq x_i$ (i.e. some of $h_0$ neutral substitutions do not occur).
  \item $n-h_0-h_+$ identical letters $y_i=x_i\neq\top_i$.
  \end{enumerate}
  If a non-substitution occurs at position $i$, then the number $n$ reduces to $n-1$, but the recombination radius $r$ remains the same.  The number $k$ of letters $y_i\neq\top_i$ remains the same in the first subcase (because $x_i\neq\top_i$ was not substituted by one of $l-k$ letters $y_i=\top_i$), but it reduces to $k-1$ in the last two subcases.  The length $l$ of the remaining string is $l-1$.
  
  Second, let us consider when letters $x_i$ are substituted by any of $k$ letters $y_i\neq\top_i$.  Again, there are three subcases for such substitutions.  They can be among
  \begin{enumerate}
  \item $h_-$ letters $y_i\neq\top_i$ (i.e. some of $h_-$ possible deleterious substitutions occur).
  \item $h_0$ letters $y_i\neq\top_i$, $y_i\neq x_i$ (i.e. some of $h_0$ neutral substitutions occur).
  \item $k-h_0-h_-$ identical letters $y_i=x_i\neq\top_i$.
  \end{enumerate}
  Note that $k-h_0-h_-=n-h_0-h_+$, which follows from $h_+-h_-=n-k$.  If a substitution by letter $y_i\neq\top_i$ occurs at position $i$, then the number $k$ reduces to $k-1$, and the recombination radius $r$ reduces to $r-1$.  The number $n$ of letters $x_i\neq\top_i$ remains the same in the first subcase (because one of $l-n$ letters $x_i=\top_i$ was substituted by $y_i\neq\top_i$), but it reduces to $n-1$ in the last two subcases.  The length $l$ of the remaining string is $l-1$.

  For each of the six subcases, there are two possibilities for the second event $z_j\neq\top_j$ in the remaining string of length $l-1$, so that there are $6\times 2=12$ joint events.  Note that the values $h_+$, $h_0$ and $h_-$ are generally random and related by equations~(\ref{eq:nr-sum}) and (\ref{eq:nr-difference}).  Thus, initially we derive the formula for the product $l(l-1)\langle P_{ij}\rangle$ assuming that specific values of $h_+$, $h_0$ and $h_-$ have been fixed:
  \begin{align*}
    l(l-1)\langle P_{ij}\rangle=\left(1-\frac{r}{l}\right)\Biggl\{h_+&\left[(n-1)\left(1-\frac{r}{l-1}\right)+k\frac{r}{l-1}\right]\\
    +h_0&\left[(n-1)\left(1-\frac{r}{l-1}\right)+(k-1)\frac{r}{l-1}\right]\\
    +(n-h_0-h_+)&\left[(n-1)\left(1-\frac{r}{l-1}\right)+(k-1)\frac{r}{l-1}\right]\Biggr\}\\
    +\frac{r}{l}\Biggl\{h_-&\left[n\left(1-\frac{r-1}{l-1}\right)+(k-1)\frac{r-1}{l-1}\right]\\
    +h_0&\left[(n-1)\left(1-\frac{r-1}{l-1}\right)+(k-1)\frac{r-1}{l-1}\right]\\
    +(k-h_0-h_-)&\left[(n-1)\left(1-\frac{r-1}{l-1}\right)+(k-1)\frac{r-1}{l-1}\right]\Biggr\}
  \end{align*}
  The formula above contains six lines corresponding to six subcases of the first event $z_i\neq\top_i$: three non-substitutions of $x_i\neq\top_i$ and three substitutions of $x_i=\top_i$.  Expressions in square brackets on each line correspond to two cases of the second event $z_j\neq\top_j$, $j\neq i$: non-substitutions of $x_j\neq\top_j$ and substitutions of $x_j=\top_j$.  Factoring and noticing that $h_++h_0+n-h_0-h_+=n$ and $h_-+h_0+k-h_0-h_-=k$ we obtain:
  \begin{align*}
    l(l-1)\langle P_{ij}\rangle=\left(1-\frac{r}{l}\right)\Biggl\{n&\left[(n-1)\left(1-\frac{r}{l-1}\right)+(k-1)\frac{r}{l-1}\right]+h_+\frac{r}{l-1}\Biggr\}\\
    +\frac{r}{l}\Biggl\{k&\left[(n-1)\left(1-\frac{r-1}{l-1}\right)+(k-1)\frac{r-1}{l-1}\right] + h_-\left(1-\frac{r-1}{l-1}\right)\Biggr\}
  \end{align*}
  The right-hand-side of the above equation can be written as
  \begin{align*}
    \left(1-\frac{r}{l}\right)\Biggl\{n&\left[(n-1)\left(1-\frac{r-1}{l-1}\right)+(k-1)\frac{r-1}{l-1}\right]+n\left(\frac{k-1}{l-1}-\frac{n-1}{l-1}\right)+h_+\frac{r}{l-1}\Biggr\}\\
    +\frac{r}{l}\Biggl\{k&\left[(n-1)\left(1-\frac{r-1}{l-1}\right)+(k-1)\frac{r-1}{l-1}\right] + h_-\left(1-\frac{r-1}{l-1}\right)\Biggr\}\,,
  \end{align*}
  which allows us to factor the members as follows:
  \begin{multline*}
    \left[\left(1-\frac{r}{l}\right)n+\frac{r}{l}k\right]\left[\left(1-\frac{r-1}{l-1}\right)(n-1)+\frac{r-1}{l-1}(k-1)\right]+\\
     +\left(1-\frac{r}{l}\right)\left(\frac{n}{l-1}(k-n) + h_+\frac{r}{l-1}\right) + \frac{r}{l}h_-\left(1-\frac{r-1}{l-1}\right)\,.
  \end{multline*}
  Substituting $h_-=h_++(k-n)$ the expression for $l(l-1)\langle P_{ij}\rangle$ becomes
  \begin{multline*}
    \left[\left(1-\frac{r}{l}\right)n+\frac{r}{l}k\right]\left[\left(1-\frac{r-1}{l-1}\right)(n-1)+\frac{r-1}{l-1}(k-1)\right]+\\
    +\left(1-\frac{r}{l}\right) \frac{(n+r)(k-n)+2rh_+}{l-1}\,.
  \end{multline*}
  It can now be combined with the expression for $l\langle P_i\rangle$ in equation $l\langle P_i\rangle+l(l-1)\langle P_{ij}\rangle-(l\langle P_i\rangle)^2$ to derive the variance formula with fixed $h_+$:
  \begin{align*}
    \sigma_P^2\{m\mid n,k,h,r,h_+\}&=\frac{l2h_+-k^2+lk+2nk-ln-n^2}{l^2(l-1)}r(l-r)\\
                                 &=\frac{l[2h_+ - (n - k)] - k^2 + 2nk - n^2}{l^2(l-1)}r(l-r)\\
                                 &=\left[2h_+ - (n - k) - \frac{(n - k)^2}{l}\right]\frac{r(l-r)}{l(l-1)}\\
                                 &=\left[h-h_0 - \frac{(n-k)^2}{l}\right]\frac{r(l-r)}{l(l-1)}\,.
  \end{align*}
  Formula~(\ref{eq:var-recombination}) for the variance is obtained by averaging over all possible values of $h_+$ or $h_0=h-2h_++n-k$, which means that they are replaced by their expected values $\langle h_+\rangle$ or $\langle h_0\rangle:=\bE\{h_0\mid n,k,h\}$ with respect to $P(h_+\mid n,k,h)$~(\ref{eq:p-h+}).
\end{proof}

\subsection{Proof of Proposition~\ref{pr:symmetry}}\label{sec:proof-symmetry}

\begin{proof}
  Looking at factorization~(\ref{eq:p-recombination-marginal}) of probability $P(m\mid n,k,h,r)$, one can see that probability $P(h_+\mid n,k,h)$ does not influence the result, because it does not include variables $m$ and $r$.  Thus, we only need to consider probability $P(m\mid n,k,h,r,h_+)$ given by equation~(\ref{eq:p-intersection-h+}).  Using symmetry of binomial coefficients $\binom{l}{r}=\binom{l}{l-r}$, one can see that
  \[
    \binom{l-h_+-h_-}{r-r_+-r_-}\binom{h_-}{r_-}\binom{h_+}{r_+}=
    \binom{l-h_+-h_-}{r'-r'_+-r'_-}\binom{h_-}{r'_-}\binom{h_+}{r'_+}\,,
  \]
  where $r'=l-r$, $r'_+=h_+-r_+$ and $r'_-=h_--r_-$.  Note that $r'=l-r$ is recombination radius of the dual recombination $z'=(1-\lambda)y\oplus\lambda x$, and $r'=r'_++r'_-+r'_0$.  Therefore, $P(m\mid n,k,h,r,h_+)=P(m'\mid n,k,h,l-r,h_+)$, where $m'=n+k-m$ by Proposition~\ref{pr:dual}.
\end{proof}

\end{appendices}

\end{document}